%% file: main.tex
\documentclass[11pt,letterpaper]{article}
\usepackage[margin=1in]{geometry}
\usepackage{datetime}
\usepackage{enumitem}
\usepackage{wrapfig}
\usepackage{subcaption}
\usepackage[T1]{fontenc}
\usepackage[utf8]{inputenc}
\usepackage{pgfplots}
\pgfplotsset{compat=1.15}
\usepackage{amsmath, amssymb}
\usepackage{amsthm}
\usepackage{color}
\usepackage{cases}
\usepackage{xparse}
\usepackage{xargs}
\usepackage{appendix}
\usepackage{verbatim, xspace}
\usepackage{tikz}
\usepackage{mathtools}
\usepackage[font=footnotesize]{caption}

\usepackage{framed}
\usepackage[framemethod=TikZ]{mdframed}

\usepackage[linesnumbered,lined,boxed,commentsnumbered,noend]{algorithm2e}

\usepackage[hidelinks]{hyperref}
\usepackage{pgfplots}

\usepackage{xcolor}

\usepackage[numbers]{natbib}
\usepackage[colorinlistoftodos,prependcaption,textsize=tiny]{todonotes}
\newcommandx{\unsure}[2][1=]{\todo[linecolor=green,backgroundcolor=green!25,bordercolor=green,#1]{\normalsize #2}}
\newcommandx{\improvement}[2][1=]{\todo[inline,linecolor=blue,backgroundcolor=blue!05,bordercolor=blue,#1]{\normalsize #2}}
\newcommandx{\info}[2][1=]{\todo[linecolor=yellow,backgroundcolor=yellow!25,bordercolor=yellow,#1]{#2}}
\newcommandx{\floatmodel}[2][1=]{\todo[inline,linecolor=red,backgroundcolor=yellow!25,bordercolor=yellow,#1]{#2}}
\newcommandx{\thiswillnotshow}[2][1=]{\todo[disable,#1]{#2}}
\newcommandx{\karol}[2][1=]{\todo[inline,linecolor=blue,backgroundcolor=blue!25,bordercolor=blue,caption={\normalsize \textbf{Karol}},#1]{\normalsize #2}}
\newcommandx{\anka}[2][1=]{\todo[inline,linecolor=red,backgroundcolor=red!25,bordercolor=red,caption={\normalsize \textbf{Anka}},#1]{\normalsize #2}}
\newcommandx{\mateusz}[2][1=]{\todo[inline,linecolor=magenta,backgroundcolor=magenta!25,bordercolor=magenta,caption={\normalsize \textbf{Mateusz}},#1]{\normalsize #2}}\newcommandx{\jedrzej}[2][1=]{\todo[inline,linecolor=gray,backgroundcolor=red!25,bordercolor=red,caption={\normalsize \textbf{Jędrzej}},#1]{\normalsize #2}}
\newcommandx{\michal}[2][1=]{\todo[inline,linecolor=gray,backgroundcolor=red!25,bordercolor=red,caption={\normalsize \textbf{Michał}},#1]{\normalsize #2}}
\newcommandx{\mipi}[2][1=]{\todo[linecolor=gray,backgroundcolor=red!25,bordercolor=red,caption={\normalsize \textbf{Michał}},#1]{\normalsize #2}}

\newtheorem{theorem}{Theorem}[section]
\newtheorem{definition}[theorem]{Definition}
\newtheorem{lemma}[theorem]{Lemma}
\newtheorem{corollary}[theorem]{Corollary}
\newtheorem{claim}[theorem]{Claim}
\newtheorem{proposition}[theorem]{Proposition}

\newtheorem{conjecture}[theorem]{Conjecture}

\newcommand{\thistheoremname}{}
\newtheorem*{genericthm}{\thistheoremname}

\newcommand{\Oh}{\mathcal{O}}

\newcommand{\nat}{\mathbb{N}}
\newcommand{\N}{\mathbb{N}}

\newcommand{\Ss}{\mathcal{S}}

\newcommand{\Tt}{\mathcal{T}}

\newcommand{\poly}{\mathrm{poly}}

\newcommand{\ol}{\overline}

\newcommand{\dist}{\mathrm{dist}}

\renewcommand{\leq}{\leqslant}
\renewcommand{\geq}{\geqslant}
\renewcommand{\le}{\leqslant}
\renewcommand{\ge}{\geqslant}

\renewcommand{\phi}{\varphi}

\newcommand{\ED}{\mathsf{ed}}
\newcommand{\FO}{\mathsf{FO}}
\newcommand{\MSO}{\mathsf{MSO}}

\newcounter{openquestion}

\newcommand{\defproblem}[3]{
	\vspace{2mm}
	\vspace{1mm}
	\noindent\fbox{
		\begin{minipage}{0.95\textwidth}
			#1 \\
			{\bf{Input:}} #2  \\
			{\bf{Task:}} #3
		\end{minipage}
	}
	\vspace{2mm}
}

\title{Dynamic data structures for parameterized string problems\thanks{This
work is a part of projects {\sc{BOBR}} (MP), {\sc{TIPEA}} (KW), and
{\sc{CUTACOMBS}} (AZ-P) that have received funding from the European Research
Council (ERC) under the European Union's Horizon 2020 research and innovation
programme (grant agreements no. 948057, 850979 and 714704, respectively).}}

\author{
    J\k{e}drzej Olkowski\footnote{Faculty of Mathematics, Informatics, and Mechanics, University of Warsaw, Poland, \texttt{jo417777@students.mimuw.edu.pl}}
    \and
    Micha\l{} Pilipczuk\footnote{Institute of Informatics, University of Warsaw, Poland, \texttt{michal.pilipczuk@mimuw.edu.pl}}
    \and
    Mateusz Rychlicki\footnote{Faculty of Mathematics, Informatics, and Mechanics, University of Warsaw, Poland, \texttt{m.rychlicki@students.mimuw.edu.pl}}
    \and
    Karol W\k{e}grzycki\footnote{Saarland University and Max Planck Institute for Informatics,
        Saarbr\"ucken, Germany, \texttt{wegrzycki@cs.uni-saarland.de}
    }
    \and
    Anna Zych-Pawlewicz\footnote{Institute of Informatics, University of Warsaw, Poland, \texttt{a.zych@mimuw.edu.pl}}
}

\date{}

\begin{document}

\hypersetup{pageanchor=false}
\begin{titlepage}
\maketitle
\thispagestyle{empty}

\input{chapters/abstract.tex}

\begin{picture}(0,0)
\put(-70,-180)
{\hbox{\includegraphics[width=40px]{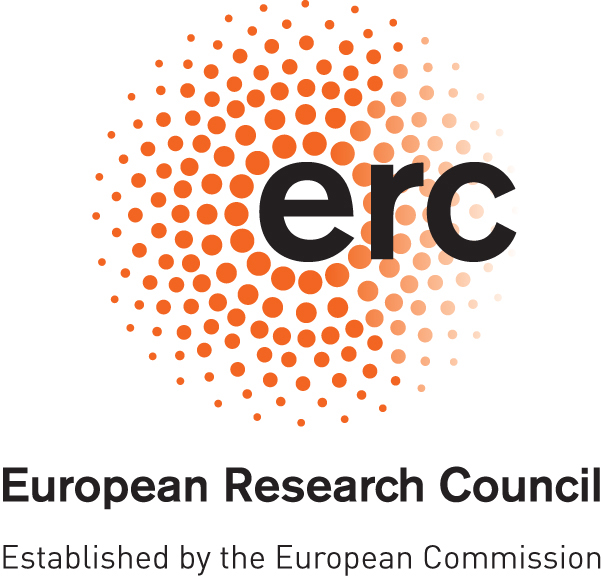}}}
\put(-80,-240)
{\hbox{\includegraphics[width=60px]{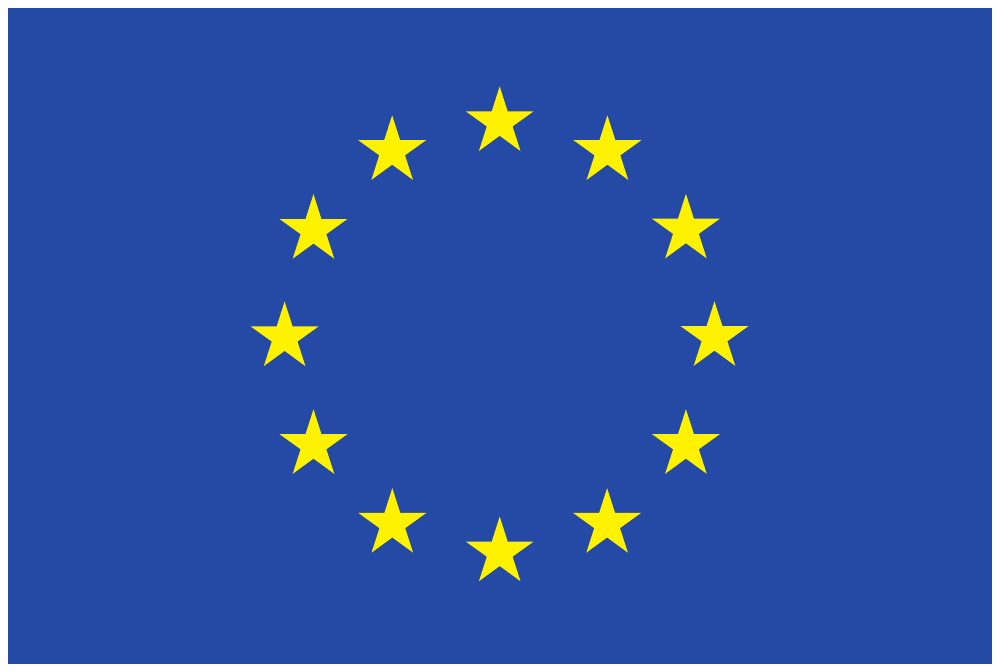}}}
\end{picture}

\end{titlepage}

\hypersetup{pageanchor=true}

\input{chapters/introduction.tex}

\subparagraph*{Organization}

In the Section~\ref{sec:prelim} we give a short preliminaries. Next, in
Section~\ref{sec:closest-string} we present a proof of
Theorem~\ref{thm:ClosestString-main} and in Section~\ref{sec:meta-theorem} we
prove Theorem~\ref{thm:intro-meta}. Subsequently, in Section~\ref{sec:examples}
we give dynamic data structures for {\sc{Disjoint Factors}} and {\sc{Edit
Distance}} and in Section~\ref{sec:lower-bounds} we show lower bounds for them.
Appendix~\ref{sec:maintain-phi} contains omitted proofs.

\input{chapters/prelim.tex}
\input{chapters/closest-string.tex}

\input{chapters/closest-string-alphabet.tex}
\input{chapters/meta-theorems}

\input{chapters/examples.tex}
\input{chapters/lower-bounds.tex}

\bibliographystyle{abbrv}
\bibliography{bib}

\begin{appendix}
    \input{chapters/appendix.tex}
\end{appendix}

\end{document}

%% file: chapters/abstract.tex
\begin{abstract}
    We revisit classic string problems considered in the area of parameterized complexity, and study them through the lens of dynamic data structures. That is, instead of asking for a static algorithm that solves the given instance efficiently, our goal is to design a data structure that efficiently maintains a solution, or reports a lack thereof, upon updates in the instance.
    
    We first consider the {\sc{Closest String}} problem, for which we design
    randomized dynamic data structures with amortized update times $d^{\Oh(d)}$
    and $|\Sigma|^{\Oh(d)}$, respectively, where $\Sigma$ is the alphabet and $d$ is the assumed bound on the maximum distance. These are obtained by combining known static approaches to {\sc{Closest String}} with color-coding.
    
    Next, we note that from a result of Frandsen et al.~[J. ACM'97] one can
    easily infer a meta-theorem that provides dynamic data structures for
    parameterized string problems with worst-case update time of the form
    $\Oh_k(\log \log n)$, where $k$ is the parameter in question and $n$ is the
    length of the string. We showcase the utility of this meta-theorem by giving
    such data structures for problems {\sc{Disjoint Factors}} and {\sc{Edit
    Distance}}. We also give explicit data structures for these problems, with
    worst-case update times $\Oh(k2^{k}\log \log n)$ and $\Oh(k^2\log \log n)$,
    respectively. Finally, we discuss how a lower bound methodology introduced
    by Amarilli et al.~[ICALP'21] can be used to show that obtaining update time
    $\Oh(f(k))$ for {\sc{Disjoint Factors}} and {\sc{Edit Distance}} is unlikely
    already for a constant value of the parameter $k$.
\end{abstract}

%% file: chapters/introduction.tex
\section{Introduction}

The field of parameterized complexity is based on the principle of {\em{parameterization}}: measuring the usage of resources not only in terms of the total input size, but also in terms of auxiliary complexity measures called {\em{parameters}}. Traditionally, the principle is applied to static algorithms and their running times, but the idea can be --- and has been --- used within essentially every algorithmic paradigm. Among these, a recent line of research has identified the area of dynamic data structures as one where the application of the parameterized approach leads to new and interesting results, see e.g.~\cite{AlmanMW20,ChenCDFHNPPSWZ21,DvorakKT14,DvorakT13,GrezMPPR21,MajewskiPS21}. In this work we continue this promising direction by investigating classic string problems considered in parameterized complexity.

Arguably, the most widely known parameterized string problem is {\sc{Closest String}}.

\defproblem{{\sc{Closest String}}}{Integer $d$ and words $s_1,s_2,\ldots,s_n\in \Sigma^L$ over an alphabet $\Sigma$, each of length $L$}{Decide whether there exists a word $c\in \Sigma^L$ such that for every $i\in \{1,\ldots,n\}$, the Hamming distance between $s_i$ and $c$ is at most $d$.}

{\sc{Closest String}} has several natural parameters: $n$, $d$, $L$, $|\Sigma|$. For the parameterization by $d$ and $\Sigma$, Gramm et al.~\cite{closest-string-03} gave a $d^{\Oh(d)}\cdot (nL)^{\Oh(1)}$-time algorithm, while Ma and Sun~\cite{MaS09} gave a $|\Sigma|^{\Oh(d)}\cdot (nL)^{\Oh(1)}$-time algorithm. By now, these are literally textbook examples of the technique of branching~\cite[Theorem~3.14 and Exercise~3.25]{the-book}, and their running times are known to be asymptotically optimal under the Exponential Time Hypothesis (ETH)~\cite{LokshtanovMS18}. For the parameterization by $n$, the classic algorithm of Gramm et al.~\cite{closest-string-03} solves the problem in time $2^{n^{\Oh(n)}}\cdot L^{\Oh(1)}$ by a reduction to integer programming in dimension $n^{\Oh(n)}$. Recently, Kouteck\'y et al.~\cite{KnopKM20} improved this running time to $n^{\Oh(n^2)}\cdot L^{\Oh(1)}$ using exciting developments in parameterized algorithms for block-structured integer programs. Kernelization algorithms for {\sc{Closest String}} were studied in~\cite{BasavarajuP0R018}.

We study the {\em{dynamic variant}} of {\sc{Closest String}}, which is to design a dynamic data structure supporting the following operations:
\begin{itemize}
\item {\em{Initialize}} the data structure for a given instance of {\sc{Closest String}}.
 \item {\em{Update}} the data structure upon modification of a single symbol in a single string $s_i$.
 \item {\em{Query}} whether the current instance is a yes-instance of {\sc{Closest String}}.
\end{itemize}
Note that parameters $n$, $d$, $L$, and $\Sigma$ are fixed on the initialization and do not change over the life of the data structure; only the strings $s_1,\ldots,s_n$ can be modified, and by one symbol at the time. Also, we assume that upon query, the data structure is only required to answer {\em{yes}} or {\em{no}}, and does not need to provide the solution $c$.

For this variant we give randomized dynamic data structures whose update time match the parametric factors in the runtimes of the algorithms of Gramm et al.~\cite{closest-string-03} and of Ma and Sun~\cite{MaS09}.

\begin{theorem}\label{thm:ClosestString-main}
 The dynamic variant of {\sc{Closest String}} admits a randomized data structure
 with initialization time $2^{\Oh(d)}\cdot nL|\Sigma|^{1+o(1)}$, amortized update time $2^{\Oh(d)}$, and
 worst-case query time $d^{\Oh(d)}$ or $|\Sigma|^{\Oh(d)}$, whichever is
 smaller. The answer to each query may result with a false positive with
 probability at most $2^{-\Omega(d)}$; there are no false negatives.
\end{theorem}

In the proof of Theorem~\ref{thm:ClosestString-main} we combine the classic
approach to {\sc{Closest String}}, originating
in~\cite{closest-string-03,MaS09}, with an interesting application of
color-coding. The randomization comes from the color-coding; we can dispose of
it using standard derandomization techniques (see~\cite[Section~5.6]{the-book}),
but at the cost of introducing an additional $\Oh(\log (nL))$ factor in the update time. Also, note that by the results of~\cite{LokshtanovMS18}, under ETH one cannot expect to improve the query time to $d^{o(d)}$ or $|\Sigma|^{o(d)}$, even in the amortized~sense.

\medskip

Next, we turn attention to other problems. First, we note that combining a result of Frandsen et al.~\cite{FrandsenMS97} on dynamic word problem for aperiodic semigroups with the classic Sch\"utzenberger-McNaughton-Papert Theorem~\cite{McNaughtonP71,Schutzenberger65a} yields the following meta-theorem.

\begin{theorem}\label{thm:intro-meta}
 Suppose $\Sigma$ is a finite alphabet and $L\subseteq \Sigma^\star$ is a language definable in logic $\FO[\Sigma,<]$. Then there exists a data structure that for a given word $w\in \Sigma^\star$, which can be updated over time by replacing single symbols, maintains whether $w\in L$. The data structure can be initialized on a given word $w$ in time $\Oh(n)$ where $n=|w|$, and then every update takes worst-case time $\Oh(\log \log n)$. 
\end{theorem}

Theorem~\ref{thm:intro-meta} follows immediately from the combination explained above, so we consider it an essentially known result (though we could not find this precise formulation in the literature). What is new is the observation that this result is a very convenient tool for obtaining dynamic data structures in the parameterized setting. We showcase this by considering the following two problems.

\defproblem{{\sc{Disjoint Factors}}}{A word $w \in \{1,\ldots,k\}^{\star}$, where $k$ is an integer}{Decide whether there exists pairwise disjoint subwords $w_{1}, w_{2}, \dots , w_{k}$ of $w$ such that for each $i\in \{1,\ldots,k\}$, $w_i$ has length at least $2$ and begins and ends with symbol $i$.}

\defproblem{{\sc{Edit Distance}}}{Integer $k$ and two words $u,v\in \Sigma^\star$, where $\Sigma$ is an alphabet}{Decide whether $\ED(u, v) \leq k$, that is, whether $v$ can be obtained from $u$ by a sequence of at most $k$ {\em{edits}}, each consisting of a deletion, insertion, or substitution of a single symbol.}

{\sc{Disjoint Factors}} has been introduced in~\cite{BodlaenderTY11} as a
stepping stone for kernelization hardness of the {\sc{Disjoint Cycles}} and
{\sc{Disjoint Paths}} problems. We choose to use it in this work as an example,
because its simple combinatorial structure makes many basic ideas clearly
visible. On the other hand, {\sc{Edit Distance}} is a problem of immense
importance with multiple applications. It can be solved in time $\Oh(n^2)$ by
standard dynamic programming (where $n$ is the total length of the words). The
best currently known algorithm for {\sc{Edit Distance}} runs in
$\Oh(n^2/\log(n)^2)$ time~\cite{best-edit-distance} and under the Strong ETH,
there is no strongly subquadratic
algorithm~\cite{BackursI18,edit-distance-lb2,edit-distance-lb3,edit-distance-lb4}.
Here, we focus on parameterization by the size of the solution $k$. In terms of
this parametrization {\sc{Edit
Distance}} can be solved in $\Oh(n+k^2)$ by the celebrated Landau and Vishkin
algorithm~\cite{landau-vishkin} and even in sublinear time when
approximation is allowed~\cite{batu,AndoniN20,GoldenbergKKS21}.

We observe that both for {\sc{Disjoint Factors}} and for {\sc{Edit Distance}}, the language of yes-instances can be defined in $\FO[\Sigma,<]$ using a sentence of length bounded in terms of the parameters. Therefore, by simply applying Theorem~\ref{thm:intro-meta}, we obtain data structures for the dynamic variants of {\sc{Disjoint Factors}} and {\sc{Edit Distance}} (defined similarly as for {\sc{Closest String}}) with worst-case update times $\Oh_k(\log \log n)$. As usual with meta-theorems, the parametric dependence in these complexity guarantees is not explicit. For this reason, we also design explicit data structures for both problems.

\begin{theorem}\label{thm:DF-intro}
 The dynamic variant of {\sc{Disjoint Factors}} admits a data structure with initialization time $\Oh(k 2^k + kn)$, worst-case query time $\Oh(1)$, and worst-case update time $\Oh(k2^k \log \log n)$.
\end{theorem}

\begin{theorem}\label{thm:ED-intro}
 The dynamic variant of {\sc{Edit Distance}} admits a data structure with initialization time $\Oh(kn)$, worst-case query time $\Oh(1)$, and worst-case update time $\Oh(k^2 \log \log n)$
\end{theorem}
Theorems~\ref{thm:DF-intro} and~\ref{thm:ED-intro} are based on Lemmas~\ref{lem:dynamic-df} and~\ref{lem:ed} and the static algorithms shown under these lemmas.
Our key component are van Emde Boas trees~\cite{van-emde-boas}. This is not
surprising, as van Emde Boas trees are also the main tool underlying the proof
of Theorem~\ref{thm:intro-meta} (see~\cite{FrandsenMS97}).  In both cases, we
heavily build upon known static algorithms~\cite{landau-vishkin,BodlaenderTY11}.
We point out that these results serve mainly as a demonstration that one can
improve the dependence on the parameter guaranteed by
Theorem~\ref{thm:intro-meta} for concrete problems. We are not aware of any
previous works on {\sc{Disjoint Factors}} in exactly this dynamic setting.
{\sc{Edit Distance}} was considered in the dynamic setting for
unbounded values of $k$ and only polynomial in $n$ updates are known~\cite{ded1,ded2,ded3}.

Finally, we observe that we can use the hardness methodology proposed by Amarilli et al.~\cite{AmarilliJP21} to establish conditional lower bounds against improving the update time  in Theorems~\ref{thm:DF-intro} and~\ref{thm:ED-intro}. More precisely, we prove that already for constant values of the parameters, the problems {\sc{Disjoint Factors}} and {\sc{Edit Distance}} are {\em{prefix-$U_1$ hard}}, which means that finding a data structure for them is at least as hard as designing a data structure for the following problem: for a dynamic word $w$ over $\{0,1\}^\star$, support queries of the form ``given $i$, is the first symbol $1$ in $w$ at position $\leq i$''. Amarilli et al.~\cite{AmarilliJP21} conjectured that no data structure for this problem achieves update time~$\Oh(1)$, and our reduction carry this hardness over to the dynamic variants of {\sc{Disjoint Factors}} and {\sc{Edit Distance}}. Let us point out that the two discussed problems are just examples, and the obtained hardness methodology can be applied to a multitude of other dynamic string problems.

%
%
%
%
%
%
%

%% file: chapters/prelim.tex
\section{Preliminaries}
\label{sec:prelim}

For a parameter $\ell$, we write $\Oh_\ell(\cdot)$ to hide factors depending
only on $\ell$. The $\poly(n_1,n_2)$ denotes $(n_1n_2)^{\Oh(1)}$.  We use a
shorthand notation $[n] \coloneqq \{1,\ldots,n\}$. For two sets $X,Y$, $X
\triangle Y$ denotes their symmetric difference $(X \setminus Y) \cup (Y
\setminus X)$.  For two words $u,v \in \Sigma^L$, by $\dist(u,v)$ we denote the
Hamming distance between $u$ and $v$. For a word $u \in \Sigma^L$ and a set $X
\subseteq [L]$, we write $a[X] \in \Sigma^{|X|}$ for the word obtained from $u$
by removing all positions outside of $X$.  For $1 \leq i \leq j \leq m$, we
write $u[i:j] \in \Sigma^{j-i+1}$ for $u[\{i,\ldots,j\}]$.

\subparagraph*{Computation Model} In this paper we work in the standard word-RAM model. In all our results the
$\Oh(\log\log n )$ factors come exclusively from application of van Emde
Boas trees that solve {\sc{Predecessor}} problem, where one needs to maintain a
set $S$ of $n$ $w$-bit integers. In update one can insert/delete integers
to/from set $S$. During query, for a given integer $x$ one should returns the largest integer $y
\in S$ such that $x \ge y$. {\sc{Predecessor}} problem is a well-studied problem
both in terms of lower and upper bounds (see the recent survey~\cite{NavarroR20}).
In word-RAM the complexity of {\sc{Predecessor}} operations is well understood
to be

\begin{equation}
    \label{eq:pred}
    \Theta\left(
        \max \left[1,
            \min \left\{
                \log_w(n),
                \frac{\log\frac{w}{\log{w}}}{\log\left(\log\frac{w}{\log{w}}/\log
                \frac{\log n}{\log w}\right)},
                \log \frac{\log(2^w-n)}{\log w}
            \right\}
        \right]
    \right)
\end{equation}

The upper and lower bounds were given by P{\u{a}}tra{\c{s}}cu and
Thorup~\cite{PatrascuT14}, see also~\cite{BeameF02,FredmanW93}. This means that strictly speaking $\Oh(\log\log n )$
factors in our paper, could be replaced with Equation~\ref{eq:pred} in word-RAM
model depending on word size. We are using the worse $\Oh(\log\log n )$ bound in order to keep the results transparent.  Note that the $\Oh(\log\log n )$ bound for
{\sc{Predecessor}} is tight in more restricted computation models (see, e.g., \cite{MehlhornNA88}).

%% file: chapters/closest-string.tex
\section{Closest String}
\label{sec:closest-string}

In this section, we show the first half of Theorem~\ref{thm:ClosestString-main} by proving the following theorem.

\begin{theorem}
    \label{thm:closest-substring}
 The dynamic variant of {\sc{Closest String}} admits a randomized data structure
 with initialization time $2^{\Oh(d)} nL$, amortized update time $2^{\Oh(d)}$,
 and worst-case query time  $d^{\Oh(d)}$. The answer to each query may result
 with a false positive with
 probability at most $2^{-\Omega(d)}$; there are no false negatives.
\end{theorem}

Throughout this section we fix the parameter $d\in \N$ and denote $\Ss\coloneqq\{s_1,\ldots,s_n\}$ for brevity, and call it a {\em{dictionary}}. Then updates on such a dictionary consist of replacing one symbol in one word with another symbol.
Our data structure is based on the static algorithm for {\sc{Closest String}} due
to Gramm et al.~\cite{closest-string-03}.

\subsection{Branching for Closest String}

\let\oldnl\nl
\newcommand{\nonl}{\renewcommand{\nl}{\let\nl\oldnl}}

\begin{algorithm}[htp]
  \SetAlgoLined\DontPrintSemicolon
  \SetKwFunction{algo}{ClosestString}
  \SetKwFunction{proc}{ClosestStringRec}
    \SetKwProg{myalg}{Algorithm}{}{}
    \SetKwProg{myproc}{Procedure}{}{}

\nonl \myalg{\algo{$\Ss,d$}}{
    \If{there exist $s_i, s_j \in \Ss$ such that $\dist(s_i,s_j) > 2d$ \label{alg:line-pair}}{\Return False}
    Set $q$ to be any word from $S$\label{l:initq}\\
    \Return $\mathtt{ClosestStringRec}$($\Ss, q, d$)\\
 }
 \nonl\;
\nonl  \myproc{\proc{$\Ss,q,x$}}{
    \If {$x < 0$}{\Return False}
        \If{there exists $s \in \Ss$ such that
            $\dist(s,q) > d$ \label{alg:line-p}}{
        Find $P \coloneqq \{ i \in [L] \text{ such that }  s[i] \neq q[i]
        \}$\tcp*{Observe that $|P| \le 3d$}
        \For {$i \in P$}{
            Set $q'[j] \coloneqq  
            \begin{cases}
                s[i] & \text{if } i=j, \\
                q[j] & \text{otherwise}.\\
            \end{cases}
            $\label{l:updateq} \\
            \If{$\mathtt{ClosestString}(\mathcal{S},q',x-1)$}{
                \Return True
            }
        }
        \Return False
    }
    \Return True
    }
    \caption{Pseudocode of static $\Oh((3d)^{d} \poly(n,L))$ time algorithm for {\sc{Closest String}}. To get a
    dynamic data structure, use
    Lemma~\ref{lem:far-word} to perform manipulations on $q$. }
    \label{alg:branching-cs}
\end{algorithm}

Algorithm~\ref{alg:branching-cs} presents a pseudocode for an
$(3d)^{d} \poly(n,L)$ time algorithm for {\sc{Closest String}}
loosely based on~\cite{closest-string-03}. We first check if every pair of words of $\Ss$ are
at distance at most more than $2d$ from each other; otherwise, by triangle inequality, we can safely terminate and return that there is no solution. Following this, we run a recursive search that maintains a candidate $q$ for a
solution, together with an upper bound $x$ on how far from $q$, in terms of
Hamming distance, we allow the sought solution to be. These are initially set to
be any word in $\Ss$ and $d$. 
Within the search, we first verify whether $q$ is already a solution. If yes, then we can terminate, this time yielding a positive answer. Otherwise, there is some $s\in
\Ss$ at distance more than $d$ from $q$. Observe that due to the initial check
and the fact that during recursion we modify $q$ at most $d$ times, it will be
always the case that $s$ and $q$ differ on at most $3d$ positions. Hence, we can
branch over one of at most $3d$ possibilities of modifying $q$ by a single letter so that $q$ gets closer to $s$. The nontrivial observation is that if there exists a solution, one of the modifications will take us closer to it in terms of the Hamming distance.

For the running time, observe that in each call we can make at most $|P| \le 3d$
guesses. Moreover, through the execution of the algorithm we can only modify at
most $d$ letters in
$q$. This means that the total size of the recursion tree is
$\Oh((3d)^{d})$.\footnote{With clever optimizations, one can decrease the running time
to be $\Oh((d+1)^d \poly(n,L))$~\cite{the-book}.}

Let us take a closer look at the polynomial factors of the algorithm presented above and discuss problems with dynamization.
In line~\ref{alg:line-pair} we need to
check if there exist words $s_i,s_j \in \Ss$ with $\dist(s_i,s_j) >
2d$. Naively, one needs to iterate over every pair of words in $\Ss$ and compute the distance
exactly which already requires $n^2$ iterations, where $n=|\Ss|$. Even if somehow, this number
could be decreased, observe that in order to compute a distance between a fixed pair
of words one needs to at least read them in $\Oh(L)$ time, which is too slow. Later, manipulations on the candidate word $q$ also require $\Oh(nL)$ time in each call of the recursive procedure $\mathtt{ClosestStringRec}$(), as $q$ is checked against all words in~$\Ss$.

We remedy these problems by introducing a data structure that maintains a dictionary $\Ss$ and provides access to all operations needed in the algorithm presented above, including efficient manipulation of the candidate $q$. This data structure is described in the following lemma.

%
%

\newcommand{\QFP}{$\mathtt{QueryFarPair}()$\xspace}
\newcommand{\RST}{$\mathtt{Reset}()$\xspace}
\newcommand{\UCs}{$\mathtt{UpdateCandidate}()$\xspace}
\newcommand{\UC}[2]{$\mathtt{UpdateCandidate}(#1,#2)$\xspace}
\newcommand{\QFW}{$\mathtt{QueryFarWord}()$\xspace}

\begin{lemma}[Far word data structure]
    \label{lem:far-word}
    There exists a randomized data structure that maintains dictionary $\Ss$ of
    words in $\Sigma^L$ with amortized $2^{\Oh(d)}$ time updates; the initialization time is $2^{\Oh(d)} nL|\Sigma|$. The data structure provides the following method:
    \begin{itemize}[nosep]
     \item \QFP: Decide if there exist $s,s'\in \Ss$ with $\dist(s,s')>2d$. The query may also return a positive answer in case there are no $s,s'$ as above, but then it is guaranteed that the answer to the instance $(\Ss,d)$ is negative.  
    \end{itemize}
    Further, the data structure provides access to a special word $q\in \Sigma^L$ through the following methods:
    \begin{itemize}[nosep]
     \item \RST: Reset $q$ to the first word in $\Ss$.
     \item \UC{i}{a}: Change the $i$th position of $q$ to symbol $a$.
     \item \QFW: Query if there exists $s \in \Ss$ with $\dist(s,q)>d$, and if so, return the pointer to $s$ and the set of positions where $s$ and $q$ differ.
    \end{itemize}
    Usage of the above requires the following promises:
    \begin{itemize}[nosep]
     \item Usage of \RST must be preceded by obtaining a negative answer to \QFP.
     \item Following resetting $q$ to $s\in \Ss$ through usage of \RST, the user has to guarantee that the assertion $\dist(q,s)\leq d$ will hold at all times till the next usage of \RST. 
     \item Every update to any word in $\Ss$ resets $q$ to be undefined, so that \RST needs to be invoked again to enable operations on $q$.
    \end{itemize}
    Methods \QFP, \RST, \UCs, \QFW work in worst-case time $2^{\Oh(d)}$. Queries \QFP and \QFW return a false negative with probability  $2^{-\Omega(d)}$; there are no false positives.
\end{lemma}
    
A few remarks are in place regarding the use of randomness in the data structure of Lemma~\ref{lem:far-word}. Namely, random bits   are used solely in the initialization of the data structure, and the correctness of subsequent uses of query methods depends on those initial random bits. As a result, the events that queries return correct answers are {\em{not}} independent, meaning that the error probability cannot be improved in the standard way by repeating each query many times. Instead, one can improve the error probability by setting up multiple independent copies of the data structure of Lemma~\ref{lem:far-word}.

With Lemma~\ref{lem:far-word} stated, we can show how to derive
Theorem~\ref{thm:closest-substring} from it.

\begin{proof}[Proof of Theorem~\ref{thm:closest-substring} assuming Lemma~\ref{lem:far-word}]
    We initialize and maintain $\alpha\log d$ independent copies of the data structure provided by Lemma~\ref{lem:far-word} for some large enough constant $\alpha$, to be determined later. Each update and each query is accordingly relayed to all these data structures; the output of a query is the disjunction of outputs provided by the individual data structures. In this way, we may assume that we have one instance of the data structure of Lemma~\ref{lem:far-word} where the probability of a false negative is reduced to $(2^{-\Omega(d)})^{\alpha\log d}=(d^{-\Omega(d)})^\alpha$. The cost for this is that the running times of all methods are increased by a multiplicative factor of $\Oh(\log d)$; this will be immaterial in the forthcoming complexity analysis.

    It remains to implement the query: we look for a word $c$ that is at Hamming distance at
    most $d$ from all the words in $\Ss$. The idea is to run Algorithm~\ref{alg:branching-cs} with all operations replaced by suitable invocations of methods of the data structure of Lemma~\ref{lem:far-word}. Lines~\ref{alg:line-pair} and~\ref{l:initq} are replaced by invocations of methods \QFP and \RST, respectively.
    In line~\ref{alg:line-p}, we invoke method \QFW. Finally, in line~\ref{l:updateq} we use one \UCs operation before recursing, and we roll-back this update (using the \UCs method again) when returning from the recursion. The running time and the correctness (assuming no false negatives from the data structure of Lemma~\ref{lem:far-word}) follow from the
    correctness of the original static algorithm and
    Lemma~\ref{lem:far-word}. It is also easy to verify that the promises required by the data structure of Lemma~\ref{lem:far-word} are kept.

    It remains to bound the probability of a false positive. Clearly, a false positive might arise only if some invocation of a method of the data structure of Lemma~\ref{lem:far-word} returns a false negative. Since the recursion tree of procedure $\mathtt{ClosestString}$() has depth at most $d$ and branching at most $3d$, it has at most $2(3d)^d$ nodes, hence in total there are at most $1+2(3d)^d$ invocations of methods of the data structure of Lemma~\ref{lem:far-word}. By setting $\alpha$ large enough, we have $(1+2(3d)^d)\cdot (d^{-\Omega(d)})^\alpha \leq 2^{-\Omega(d)}$. So by the union bound, the probability of an error is bounded by $2^{-\Omega(d)}$.
\end{proof}

Now, we discuss the technical ideas behind the proof of
Lemma~\ref{lem:far-word}. The key idea is that we can efficiently
maintain an approximate solution, as explained in the lemma below.

\begin{lemma}[Approximate {\sc{Closest String}}]
    \label{lem:base-lemma}
    There exists a data structure that maintains a dictionary $\Ss$ of words in $\Sigma^L$ with  amortized update time $\Oh(|\Sigma|)$, as well as a word $o \in
    \Sigma^L$ with the following guarantee: if the answer to the {\sc{Closest
    String}} instance $(\Ss,d)$ is positive, then $\dist(o,s) \le 4 |\Sigma|
    \cdot d$ for every $s \in \Ss$. Data structure can be initialized in
    $\Oh(nL)$ time.

    Moreover, the data structure also maintains the set $\Delta(o,s) \coloneqq \{ i \in [L] \; | \;
    o[i] \neq s[i] \}$ for every $s \in \Ss$ and, upon request, can return $\Delta(o,s)$ in time
    $\Oh(|\Delta(o,s)|)$. Finally, the data structure can check whether
    $\dist(o,s) \le 4 |\Sigma| \cdot d$ for all $s \in \Ss$ in time $\Oh(1)$.
\end{lemma}

In Section~\ref{sec:approximate} we prove
Lemma~\ref{lem:base-lemma}. Next, in Section~\ref{sec:far-pair} we use
an approach based on color coding to leverage Lemma~\ref{lem:base-lemma} to a data structure achieving the first part of Lemma~\ref{lem:far-word}: maintaining a dictionary $\Ss$ and implementing query \QFP. Adding the functionality concerning the candidate word $q$ uses similar arguments and is presented in  Section~\ref{sec:far-word}. Looking at the statement of Lemma~\ref{lem:base-lemma}, the reader might be at this point worried that this plan involves complexities dependent also on $|\Sigma|$. However, in Section~\ref{sec:far-pair} we will show how to reduce $|\Sigma|$ to $\Oh(d)$ using color coding.

\subsection{Approximate Closest String}
\label{sec:approximate}

In this section we prove Lemma~\ref{lem:base-lemma}. The main idea is to define $o \in \Sigma^L$ through an approximate majority vote for every position, maintained in a lazy fashion.
We formalize this through the following definition.

\begin{definition}[Origin Word]
    An \emph{origin word} for a dictionary $\Ss$ of words in $\Sigma^L$ is a word $o \in \Sigma^L$ such that 
    \begin{displaymath}
        |\left\{s \in \Ss \; | \; s[i] = o[i] \right\}| \ge \frac{|\Ss|}{2|\Sigma|} \text{ for every } i \in [L].
    \end{displaymath}
    We say that the origin word $o$ is \emph{good} if $\dist(o,s) \le 4 |\Sigma|
    \cdot d$ for every $s \in \Ss$.
\end{definition}

By definition, if an origin word is good, then it is a solution for the
{\sc{Closest String}} instance $(\Ss,4|\Sigma|d)$. We now show a reverse
``soundness'' implication: if some origin word is not good, then for sure there
is no solution for $(\Ss,d)$.

\begin{lemma}
    \label{prop:correctness}
    If for an instance $(\Ss,d)$ there exists an origin word that is not good,
    then the answer to $(\Ss,d)$ is negative.
\end{lemma}
\begin{proof}
For the sake of contradiction, let us assume that there exists $c \in \Sigma^L$
such that $\dist(c,s) \le d$ for every $s \in \Ss$. Moreover, there exists some
origin word $o \in \Sigma^L$ and a witness $w \in \Ss$ such that $\dist(o,w) > 4
d|\Sigma|$.

Let $C_{o,w}$ be the total count of matches between $o$ and all words
in $\Ss$ at the positions where $o$ and $w$ differ. That is,
$$ C_{o,w} \coloneqq |\{ (i, s) \in [L] \times \Ss \; \text{ such that } \;  w[i] \neq o[i] \text{ and }  s[i] = o[i]\} |.$$
Let us show a lower bound on $C_{o,w}$. Observe that for a witness $w \in
\Ss$ there are at least  $\dist(o,w)$ positions $i$ that are taken into account when
computing $C_{o,w}$.  Moreover, by the definition of origin word $o$, for every
position $i \in [L]$ at least $|\Ss|/(2|\Sigma|)$ words match $o$ on position
$i$.  Therefore,
\begin{equation}
    \label{eq:cow-lb}
     \frac{|\Ss| \cdot \dist(o,w)}{2 |\Sigma|} \leq C_{o,w}
     .
\end{equation}
On the other hand, we assumed that there exists $c \in \Sigma^L$ such that
$\dist(c,s) \leq d$ for every $s \in \Ss$. Since $w \in \Ss$, by
triangle inequality we have $\dist(s,w) \le 2d$ for every $s \in \Ss$. Hence
\begin{equation}
    \label{eq:cow-ub}
    C_{o,w} \leq 2d \cdot  |\Ss|
    .
\end{equation}
By combining~\eqref{eq:cow-lb} and~\eqref{eq:cow-ub} we conclude that
$\dist(o,w) \le 4 |\Sigma| d$, a contradiction.
\end{proof}

Next, we argue that in $\Oh(|\Sigma|)$ time we can maintain some origin
word for a given dictionary.

\begin{lemma}
    \label{prop:runtime}
    In $\Oh(nL)$ time we can initialize 
    a data structure that for a given dictionary $\Ss$ of words in $\Sigma^L$ maintains some origin word $o\in \Sigma^L$ with amortized update time $\Oh(|\Sigma|)$. The data structure also maintains the set $\Delta(o,s) \coloneqq  \{ i \in [L] \; | \; s[i] \neq o[i] \}$ for every $s
    \in \Ss$ and upon request, can return each set $\Delta(o,s)$ in time
    $\Oh(|\Delta(o,s)|)$. Finally, the data structure can check whether $o$ is good in time $\Oh(1)$.
\end{lemma}
\begin{proof}
    Upon initialization, we set $o\in \Sigma^L$ so that for every position $i\in
    [L]$, $o[i]$ is a symbol that occurs the most often among $s[i]$ for $s \in \Ss$. Clearly, $o$ constructed in this way is an origin word. We also compute the relevant sets $\Delta(o,s)$.
    
    The data structure stores the following additional data. For every position
    $i\in [L]$ and every symbol $\alpha\in \Sigma$, we maintain a counter
    indicating the number of words $s\in \Ss$ such that $s[i]=\alpha$.
	Each set $\Delta(o,s)$ is stored as a linked list (with no assumption on the order), plus there is an array of length $L$ whose $i$th entry is either null if $i\notin \Delta(o,s)$, or contains a pointer to the relevant object on the linked list representing $\Delta(o,s)$.
	Additionally, with each set $\Delta(o,s)$ we maintain its size. Additionally, we store a single counter indicating the number of words $s\in \Ss$ such that
    $|\Delta(o,s)|\geq 4|\Sigma|d$. This counter can be used to answer queries about the goodness of $o$ in time $\Oh(1)$. 
	Upon initialization, all of the above can be computed in time $\Oh(nL)$ in a straightforward~way.
    
    We now explain how the data structure behaves upon an update. Suppose position $s_j[i]$ is modified. We update the relevant counters for position $i$ and update $\Delta(o,s_j)$ accordingly.
    Next, we check whether the counter for the symbol $o[i]$ at position $i$ did not drop below $|\Ss|/(2|\Sigma|)$. If not, then $o$ remains an origin word and there is no need to change $o$. Otherwise, we modify $o[i]$ as~follows.
    
    By iterating through all words in $\Ss$, we compute the most frequent symbol among $s[i]$ for $s\in \Ss$, and we set $o[i]$ to be this symbol.  Moreover,
    we iterate over all $s\in \Ss$ and update $\Delta(o,s)$ accordingly, by adding or removing the position $i$ if needed. These operations require total time $\Oh(|\Ss|)$.

    We now argue that the amortized update time is $\Oh(|\Sigma|)$.
    By the pigeon-hole principle, when symbol $o[i]$ gets modified, it is replaced by a symbol that occurs at least
    $|\Ss|/|\Sigma|$ times on position $i$ in words from $\Ss$. Also, this is true for the symbol placed as $o[i]$ upon initialization. Therefore, before every update when $o[i]$ gets modified, there are at least $|\Ss|/(2|\Sigma|)$ updates on position $i$ where $o[i]$ was modified. We can charge the running time $\Oh(|\Ss|)$ used when modifying $o[i]$ to those previous updates, thus obtaining amortized update time $\Oh(|\Sigma|)$.
\end{proof}

Now Lemma~\ref{lem:base-lemma} follows by combining
Lemmas~\ref{prop:runtime} and~\ref{prop:correctness}.

\subsection{Detecting dissimilar words}
\label{sec:far-pair}

In this section we prove the first part of Lemma~\ref{lem:far-word}: we present a data structure that maintains the dictionary and implements the method \QFP, and for now we ignore the methods for handling $q$.

In the data structure, we will maintain hashes of all words in $\Ss$ to the binary alphabet. More precisely, upon initialization of the data structure, we uniformly at random sample a function $h\colon [L]
\times \Sigma \rightarrow \{0,1\}$ that assigns a label $0$ or $1$ to every position and symbol in the
alphabet. This function is fixed for the whole life of the data structure and stored in it. In notation, we shall use a natural lift of
$h \colon \Sigma^L \rightarrow \{0,1\}^L$ that applies $h$ position-wise. In the data structure we store, together with $\Ss$, the hashed dictionary $\widetilde{\Ss}\coloneqq \{h(s)\colon s\in \Ss\}$. Observe that upon every update to $\Ss$ we can also update $\widetilde{\Ss}$ in constant time.

We also maintain an approximate solution $o \in \{0,1\}^L$ for the dictionary $\widetilde{\Ss}$ using the data structure of
Lemma~\ref{lem:base-lemma}. Recall that we can query the data structure of
Lemma~\ref{lem:base-lemma} about whether $\dist(o,\tilde s)\leq 8d$ for all
$\tilde s\in \widetilde{\Ss}$ and if this is not the case, then we know for sure
that the instance $(\widetilde{\Ss},d)$ of {\sc{Closest String}} has a negative
answer. Note that this conclusion implies that the original instance $(\Ss,d)$ also has a negative answer.

In addition to the approximate solution $o$, the data structure of Lemma~\ref{lem:base-lemma}
provides an access to the sets $\Delta(o,\tilde{s})$ of positions where $o$ and
$\tilde{s}$
differ, for all $\tilde{s} \in \widetilde{\Ss}$.

\newcommand{\colors}{\mathsf{colors}}

Finally, we also hash positions as follows. Upon initialization, we sample
uniformly at random a function $\pi \colon [L] \rightarrow [16d]$ which maps
positions to a set of $16 d$ colors (numbers from $1$ to $16d$). Again, this function is fixed for the whole life of the data structure and stored in it.
For a word $\tilde{s} \in
\widetilde{\Ss}$, let $\colors_{o,\pi}(s)=\{ \pi(i) \; | \; i\in [L] \textrm{ and
}s[i] \neq o[i]\}$ be the set of colors assigned to the symbols in $\tilde{s}$
that are on positions where $\tilde{s}$ does not match the origin word $o\in
\{0,1\}^L$.

In the data structure we maintain, for every $C\subseteq [16d]$, the set $\Phi(C)$ defined as follows:
$$\Phi(C)=\{ \tilde{s} \in \widetilde{\Ss} \; | \; \colors_{o,\pi}(\tilde{s}) = C\}.$$
In other words, $\Phi(C)$ is the set of words from $\widetilde{\Ss}$ that
get assigned color set $C$. The next statement shows that sets $\Phi(C)$ can be maintained in $2^{\Oh(d)}$ time
per update.

\begin{lemma}
    \label{lem:maintain-phi}
    We can initialize in $2^{\Oh(d)}\cdot nL$ time a data structure that for every $C\subseteq [16d]$ maintains the set $\Phi(C)$ in amortized $2^{\Oh(d)}$ time per
    update to $\Ss$. When queried about any $C\subseteq [16d]$, the data structure in $\Oh(1)$ time either returns any
    element from $\Phi(C)$, or asserts that $\Phi(C)$ is empty.
\end{lemma}

The proof of Lemma~\ref{lem:maintain-phi} is deferred to Appendix~\ref{sec:maintain-phi}. It is rather technical and builds on the data structure of Lemma~\ref{lem:base-lemma} by additionally storing sets $\Phi(C)$ as doubly-linked lists. Every modification to $\widetilde{\Ss}$ and $o$ triggers a number of modifications to lists representing $\Phi(C)$, consisting of moving some elements from one list to another. The same amortization argument as the one used in the proof of Lemma~\ref{lem:base-lemma} shows that the amortized update time is $2^{\Oh(d)}$.

\begin{algorithm}
    \SetKwInOut{Input}{Method}
	\SetKwInOut{Output}{Output}
    \Input{\QFP}
    \DontPrintSemicolon

\If{it is not the case that $\dist(o,\tilde{s})\leq 8d$ for all $\tilde{s}\in \widetilde{\Ss}$\label{line:origin}}{
         \Return Answer to $(\Ss,d)$ is negative.
     }

    \For{every $X,Y \subseteq [16d]$ with $|X \triangle Y| > 2d$ \label{alg:line-for}}{
        \If{$\Phi(X)$ and $\Phi(Y)$ are both nonempty}{
            \Return There are words $s,s'\in \Ss$ with $\dist(s,s')>2d$.
        }
    }
    \Return It holds that $\dist(s,s')\leq 2d$ for all $s,s'\in \Ss$.
    \caption{Pseudocode for the method \QFP.}
    \label{alg:far-pair}
\end{algorithm}

We now present implementation of the query operation;  see
Algorithm~\ref{alg:far-pair} for a pseudocode.
We first check whether $\dist(o,\tilde{s})\leq 8d$ for all $\tilde{s}\in
\widetilde{\Ss}$. As argued in Lemma~\ref{prop:correctness}, if this is not the case, then we can safely
conclude that the answer to the instance $(\Ss,d)$ is negative. Otherwise,
we iterate over every pair of sets $X,Y \subseteq [16d]$ with $|X
\triangle Y| = |(X\setminus Y) \cup (Y \setminus X)| > 2d$. Then, we check whether 
both $\Phi(X)$ and $\Phi(Y)$ are nonempty. If that is the case, then (as we will argue) any pair $(\tilde{s},\tilde{s}') \in \Phi(X) \times \Phi(Y)$ satisfies $\dist(\tilde{s},\tilde{s}')>2d$, implying that the original words $s,s'\in \Ss$ also satisfy $\dist(s,s')>2d$. Otherwise, if for every such $X$ and $Y$ at
least one of $\Phi(X)$ or $\Phi(Y)$ is empty, we conclude that there is no pair $s,s'
\in \Ss$ with $\dist(s,s') > 2d$.

Because the number of pairs $X,Y \subseteq [16d]$ is $2^{\Oh(d)}$, the
query algorithm runs in $2^{\Oh(d)}$ time in total. The next lemma shows that if the algorithm finds some pair of words and reports that they are at distance larger than $2d$, then this answer is~correct.

\begin{lemma}
 Suppose Algorithm~\ref{alg:far-pair} finds a pair $X,Y\subseteq [16d]$ with $|X\triangle Y|>2d$ and $\Phi(X)\neq \emptyset$ and $\Phi(Y)\neq \emptyset$. Then there are $s,s'\in \Ss$ such that
 $\dist(s,s')>2d$.
\end{lemma}
\begin{proof}
 Consider any pair $(\tilde{s},\tilde{s}')\in \Phi(X)\times \Phi(Y)$. Observe that for every
 color $r\in X\setminus Y$, there is a position $i$ with $\pi(i)=r$ such that
 $\tilde{s}[i]\neq o[i]$ (due to $r\in X$), and we have $o[i]=\tilde{s}'[i]$
 (due to $r\notin Y$). So $\tilde{s}[i]\neq \tilde{s}'[i]$, implying
 $s[i]\neq s'[i]$. Similarly for every $r\in Y\setminus X$. Positions $i$ as
 above have to be pairwise different due to receiving different colors in $\pi$,
 so we conclude that $s$ and $s'$ differ on more than $2d$ positions.
\end{proof}

To finish the proof, it remains to analyze the success probability of Algorithm~\ref{alg:far-pair}.

\begin{lemma}
    \label{lem:prob-far-pair}
    If there exists a pair $a,b \in \Ss$ with $\dist(a,b) > 2d$,
    then Algorithm~\ref{alg:far-pair} detects such a pair with the probability at least
    $2^{-\Oh(d)}$, or concludes that the answer to the instance $(\Ss,d)$ is
    negative.
\end{lemma}

Note that in Lemma~\ref{lem:far-word} we promised error probability bounded by
$2^{-\Omega(d)}$, while Lemma~\ref{lem:prob-far-pair} provides a bound of
$1-2^{-\Oh(d)}$ on the error probability. This can be easily remedied by
maintaining $2^{\Theta(d)}$ independent copies of the data structure. This
increases the time of update and initialization by a $2^{\Oh(d)}$ factor.

\begin{proof}[Proof of
Lemma~\ref{lem:prob-far-pair}] First, we argue that after hashing the alphabet,
we still have $\dist(h(a),h(b))>2d$ with sufficiently high probability.
Let $P \subseteq \{ i \in [L] \; | \; a[i] \neq b[i] \}$ be any set of size
exactly $2d+1$ consisting of positions where $a$ and $b$ differ. 

Let $\tilde{a} = h(a)$ and $\tilde{b} = h(b)$.
First, we claim that with probability at least $2^{-\Oh(d)}$ it holds that
$\dist(\tilde{a},\tilde{b}) > 2d$. Observe that for a fixed position $i
\in P$, the probability that $h$ assigns different symbols to
$a[i]$ and to $b[i]$ is $1/2$. Since $h$ is sampled on each position $i\in [L]$ independently, the probability that this
happens for all positions in $P$ is $2^{-|P|} = 2^{-\Oh(d)}$.

From now on, let us assume that $\dist(\tilde{a},\tilde{b}) > 2d$. Moreover,
by Line~\ref{line:origin} we may assume that $\dist(o,\tilde{s})\leq 8d$ for every $\tilde{s}\in \widetilde{\Ss}$. Consider the
set
\begin{displaymath}
    \Delta_o(\tilde{a},\tilde{b}) \coloneqq \{ i \in [L]~|~\tilde{a}[i] \neq o[i] \text{ or } \tilde{b}[i] \neq o[i] \}.
\end{displaymath}

Observe that since $\dist(o,\tilde{a})\leq 8d$ and $\dist(o,\tilde{b})\leq 8d$, we have 
$k \coloneqq |\Delta_o(\tilde{a},\tilde{b})| \in [2d + 1,16d]$. Now, we claim that with
probability $2^{-\Oh(d)}$ the function $\pi$ assigns different colors to
all positions in $\Delta_o(\tilde{a},\tilde{b})$. There are $(16d)^k$ different
colorings on $\Delta_o(\tilde{a},\tilde{b})$. However, only $\binom{16d}{k}
k!$ of them assign different colors to $\Delta_o(\tilde{a},\tilde{b})$.
Therefore the probability that $\pi$ assigns different colors on
$\Delta_o(\tilde{a},\tilde{b})$ is:
\begin{displaymath}
    \Pr\left[ |\{ \pi(i) \; | \; i \in \Delta_o(\tilde{a},\tilde{b})\}| = k \right] = 
    \frac{\binom{16d}{k} k!}{(16d)^k} = \frac{(16d)}{(16d)} \cdots
    \frac{(16d - k+1)}{(16d)}  > \frac{(16d)!}{(16d)^{16d}} >
    e^{-16d}
\end{displaymath}
where the last inequality follows from the well-known bound $n! > (n/e)^n$.
Hence, the probability that $\pi$ assigns different colors to all positions
in $\Delta_o(\tilde{a},\tilde{b})$ is $2^{-\Oh(d)}$. Now, we claim that if that
indeed happens, then Algorithm~\ref{alg:far-pair} detects a suitable pair.

Let $X$ and $Y$ be sets of colors such that $\tilde{a} \in \Phi(X)$ and $\tilde{b} \in
\Phi(Y)$. It suffices to show that $|X \triangle Y| > 2d$. Observe that every
position where $\tilde{a}$ and $\tilde{b}$ differ belongs to
$\Delta_o(\tilde{a},\tilde{b})$, hence these (more than $2d$) positions receive
different colors in $\pi$. Further, for every position where $\tilde{a}$ and $\tilde{b}$ differ, the color of this position belongs to $X\triangle
Y$, for outside of positions of $\Delta_o(\tilde{a},\tilde{b})$ the words
$o,\tilde{a},\tilde{b}$ all agree. It follows that $|X\triangle Y|>2d$.
\end{proof}

\subsection{Maintaining a candidate solution}
\label{sec:far-word}

In this section we finish the proof of Lemma~\ref{lem:far-word} by implementing the operations on the candidate word $q$. This
proof builds upon the construction from Section~\ref{sec:far-pair} using the same ideas, so we only briefly discuss the additional elements that need to be maintained.

Observe that $q$ is always reset to the first word $s_1\in \Ss$ and operations on $q$ are performed under the promise that $\dist(q,s_1)\leq d$ at all times. Therefore, we maintain $q$ implicitly by remembering only at most $d$ positions on which $s_1$ and $q$ differ, and what are the symbols of $q$ on those positions. This allows us to implement the reset and update operations for $q$ in time $2^{\Oh(d)}$. (Recall here that in Section~\ref{sec:far-pair} we in fact maintained $2^{\Oh(d)}$ independent copies of the data structure in order to boost the error probability.) Also, we maintain the hashed version $\tilde{q}\coloneqq h(q)$. 

The method \QFW is implemented using a similar mechanism as was used in Section~\ref{sec:far-pair}. For technical reasons, we extend the palette of colors used by $\pi$ from $[16d]$ to $[17d]$. Then we maintain the sets $\Phi(C)\subseteq \widetilde{\Ss}$ for $C\subseteq [17d]$ as before. However, instead of iterating over all pairs $X,Y\subseteq [17d]$ with $|X\triangle Y|>2d$, we first compute $Q\coloneqq \colors_{o,\pi}(\tilde{q})$ and then iterate over all $X\subseteq [17d]$ such that $|X\triangle Q|>d$ and check whether $\Phi(X)$ is nonempty. The same reasoning as in Section~\ref{sec:far-pair} shows that if there exists $s\in \Ss$ with $\dist(q,s)>d$, then with high enough probability we will find such an $s$ as any element of $\Phi(X)$.

Note that in the description above we did not specify how the set $\colors_{o,\pi}(\tilde{q})$ is computed. This can be done by first obtaining the set $\Delta(o,\tilde{s}_1)$ from the data structure of Lemma~\ref{lem:base-lemma}, and then inspecting all the positions of $\Delta(o,\tilde{s}_1)\cup \Delta(s_1,q)$, where $\Delta(s_1,q)$ are the at most $d$ positions where $s_1$ and $q$ differ. Note here that we may assume that $|\Delta(o,\tilde{s}_1)|\leq 8d$, for otherwise the data structure presented in Section~\ref{sec:far-pair} must have returned that the answer to $(\Ss,d)$ is negative when resetting $q$. Further, this reasoning shows that $|\Delta(o,\tilde{q})|\leq 9d$ at all times. Note that in the correctness argument presented in Section~\ref{sec:far-pair} we used the assumption that $\dist(o,\tilde{a})\leq 8d$ and $\dist(o,\tilde{b})\leq 8d$, and this is why we chose a palette of colors of size $16d$. Now we have $\dist(o,\tilde{q})\leq 9d$ and $\dist(o,\tilde{s})\leq 8d$, so a palette of $17d$ colors suffices.

It remains to argue that if $s$ with $\dist(q,s)>d$ is found, the set $P$ of positions on which $q$ and $s$ differ can be reported in time $\Oh(d)$. But again, $P$ can be constructed by inspecting all positions of $\Delta(o,\tilde{s}_1)\cup \Delta(q,s_1)$, and this set has size $\Oh(d)$ and can be obtained by a query to the data structure of Lemma~\ref{lem:base-lemma}. This finishes the proof of Lemma~\ref{lem:far-word}.

%% file: chapters/closest-string-alphabet.tex
\ifx\islipics\undefined
\subsection{{\sc{Closest String}} for small alphabets}
\label{sec:closest-string-small-alphabets}
\else
\section{{\sc{Closest String}} for small alphabets}
\label{sec:closest-string-small-alphabets}
\fi

In this section we analyse the complexity of {\sc{Closest String}} for small
alphabets and show that our techniques also apply in this setting. That is, we prove the second half of Theorem~\ref{thm:ClosestString-main}, presented below.

\begin{theorem}\label{thm:small-alphabets}
     The dynamic variant of {\sc{Closest String}} admits a randomized data
     structure with initialization time $2^{\Oh(d)} nL|\Sigma|^{1+o(1)}$, amortized update time
     $2^{\Oh(d)}$, and worst-case query time  $(|\Sigma|-1)^{d} 2^{\Oh(d)}$. The answer to
     each query may result with a false negative with probability
     at most $2^{-\Omega(d)}$; there are no false positives.
\end{theorem}

The strategy is exactly the same as in Section~\ref{thm:closest-substring}.
First, we present a static algorithm with running time $(|\Sigma| - 1)^{d} \cdot 2^{\Oh(d)} \cdot (nL)^{\Oh(1)}$, which is essentially the algorithm proposed by Ma and Sun~\cite{closest-string-03}.
Next, we show how to use Lemma~\ref{lem:far-word} to implement this static algorithm
to the dynamic~setting.
The algorithm is presented using pseudocode as
Algorithm~\ref{alg:branching-cs-2}. We present it somewhat differently than Ma and
Sun in order to streamline the analysis of the dynamic
variant.
\begin{algorithm}
	\SetKwInOut{Input}{Algorithm}
	\SetKwInOut{Output}{Output}
    \Input{$\mathtt{ClosestStringSmallAlphabet}$($\mathcal{S}, d$)}
    \DontPrintSemicolon
    
    Set $F \coloneqq \emptyset$\\
    Set $q$  to be the first word $s_1 \in \Ss$\label{alg:line-p1}\\
    Set $b\coloneqq d$\\
    \While{exists $s \in \Ss$ such that
    $\dist(s,q) > d$ \label{alg:line-p2}}{
    Find $P \coloneqq \{ i \in [L] \text{ such that }  s[i] \neq q[i] \} \setminus F$\\
    \If{$\dist(s,q) > 2d$ or $P = \emptyset$\label{line:conds}}{\Return False}
    Guess $Q=\{i\in P\text{ such that }c[i]\neq q[i]\}$
	\tcp*{$c \in \Sigma^L$ denotes the sought solution}
    \If{$|Q|>b$ or $Q=\emptyset$\label{line:conds2}}{\Return False}
    \For {$i \in Q$}{ 
        Guess $c[i] \in \Sigma\setminus \{q[i]\}$ \label{line:guess}\\
        Set $q[i] \coloneqq c[i]$
    }
    $F \coloneqq F \cup P$\\
    $b \coloneqq \min(d-\dist(s,q),b-|Q|)$\\
    \If{$b<0$}{\Return False}
    }
    \Return True
    \caption{Pseudocode of a $(|\Sigma|-1)^{d} \cdot 2^{\Oh(d)} \cdot (nL)^{\Oh(1)}$-time static
    algorithm for {\sc{Closest String}}. To get a dynamic data structure with
     query time $(|\Sigma|-1)\cdot 2^{\Oh(d)}$, use Lemma~\ref{lem:far-word} for operations on $q$.}
    \label{alg:branching-cs-2}
\end{algorithm}
The algorithm maintains three global values. The first one is a set $F \subseteq
[L]$ of \emph{fixed indices}.  The second one is a  word $q \in \Sigma^L$ that
is a candidate for the solution, which at the start is set to be any word from
$\Ss$. The third one is a {\em{budget}} $b\in \N$, initially set to $d$.

We imagine the algorithm as a nondeterministic procedure that, having in
mind some solution $c\in \Sigma^L$, guesses parts of $c$ along the execution and
appropriately modifies $q$. The set $F$ is used to keep track of the positions that are already assumed to be fixed as in $F$. As usual,
nondeterministism is determined by branching over all possibilities, and the
total number of branches determines the running time of the algorithm.
At every point, even in branches where guesses were inconsistent with $c$, the algorithm maintains the following invariant:
\begin{enumerate}[label=$(\diamondsuit)$]
 \item\label{inv} There is $s\in \Ss$ such that $s[\ol{F}]=q[\ol{F}]$ and $b=d-\dist(q[F],s[F])$, where we denote $\ol{F}=[L]\setminus F$.
\end{enumerate}
In this way, one may think of $b$ as of the budget that is left for changing symbols positions in $q$ outside of $F$: at most $b$ of them can be still changed, for otherwise the solution would be too far from $s$.

Every step of the algorithm works as follows.  First, we find a word $s \in \Ss$
with $\dist(s,q) > d$. If no such word exists, then the current candidate $q$ is
a solution and we can terminate the procedure claiming a positive answer.
Otherwise, we compute the set $P$ of positions where $s$ and $q$ differ, and we
remove from it all positions that were fixed before.

Next, we check whether $\dist(s,q) > 2d$, which translates to the condition
$\dist(s[F\cup P],q[F \cup P]) > 2d$. If this is the case, we terminate
and provide a negative answer: there is no way to obtain a word at distance at most $d$ from $s$ by changing at most $d$ positions in $q$. If $P = \emptyset$, we can also terminate and provide a negative answer:
already on fixed positions, our candidate $q$ and $s$ differ by more than $d$.  Otherwise, when $\dist(s,q) \le 2d$ and $P \neq \emptyset$, we guess exactly the
symbols in $c$ at positions from $P$ and we modify $q$ to have $q[P] = c[P]$. This is done through a two-stage process: first we guess the set of positions $Q\subseteq P$ where $q$ needs to be modified, and then we guess the symbols of $c$ at positions of $Q$; for each there are $|\Sigma|-1$ possibilities. Note that we may restrict attention to sets $Q$ that are nonempty (for $\dist(s,q)>d$) and satisfy $|Q|\leq b$ (by invariant \ref{inv}). Finally,
we add $P$ to the set $F$ of fixed indices and update $b$ to the minimum of the two values: $d-\dist(s,q)$ and $b-|Q|$. It is straightforward to verify that this way invariant \ref{inv} is still maintained: either $s$ or the previous witness for \ref{inv} may serve as the new witness for \ref{inv}. Clearly, if $b$ became negative, it is safe to terminate the branch. Otherwise we continue the search until a candidate $q$ at distance at most $d$ from all strings in $\Ss$ is found.

This concludes the description of the algorithm. The correctness is clear from
the description as we return True only if our candidate is at the distance at
most $d$ from all input strings.

It is now straightforward to turn this algorithm into a dynamic data structure
just as we did in the proof of Theorem~\ref{thm:closest-substring}. Namely, we
maintain the data structure of Lemma~\ref{lem:far-word}, and use it to operate
on the candidate word $q$. All distance checks can be implemented in linear time by verifying the $\Oh(d)$-sized difference sets provided by this data structure. We will later show that the whole recursion tree of Algorithm~\ref{alg:branching-cs-2} has total size at most $(|\Sigma|-1)^d\cdot 2^{\Oh(d)}$. 
Hence, as the operations in the data structure of
Lemma~\ref{lem:far-word} take amortized time $2^{\Oh(d)}$, the complexity
guarantees promised in Theorem~\ref{thm:small-alphabets} follow in the same way
as it was the case for Theorem~\ref{thm:closest-substring}. As for the error probability, we can maintain $\alpha\cdot \log |\Sigma|$ independent copies of the data structure of Lemma~\ref{lem:far-word} for some large constant $\alpha$, so that the probability that this composite data structure returns a false negative is reduced to $(|\Sigma|^{-\Omega(d)})^{\alpha}$. Then, just as in the proof of  Theorem~\ref{thm:closest-substring}, it follows from the union bound that the probability of a false negative in Algorithm~\ref{alg:branching-cs-2} is at most $2^{-\Omega(d)}$.

We are left with bounding the running time of
Algorithm~\ref{alg:branching-cs-2}, or more precisely, showing that the whole recursion tree has size at most $(|\Sigma|-1)^d\cdot 2^{\Oh(d)}$. The argument conceptually follows the reasoning of Ma and Sun~\cite{closest-string-03}; we present it for completeness.

\subparagraph*{Runtime} The key observation is the following lemma.

\begin{lemma}
    \label{claim:runtime}
	Consider $i$th iteration of the while loop in Algorithm~\ref{alg:branching-cs-2} (with any guesses made). Let $b_i$ be the value of $b$ before this iteration, and $b_{i+1}$ be the value of $b$ after this iteration. Then $b_{i+1}\leq b_i/2$.
\end{lemma}

We first argue that the claimed runtime of Algorithm~\ref{alg:branching-cs-2} follows
from Lemma~\ref{claim:runtime}. Consider any root-to-leaf path in the recursion tree of the algorithm; this corresponds to a single run of Algorithm~\ref{alg:branching-cs-2} treated as a nondeterministic procedure, with some guesses made along the way. For iterations $i=1,2,\ldots,p$ of the while-loop, where $p$ is the total number of iterations made, let $b_i$ be the value of $b$ at the beginning of the $i$th iteration, and let $\ell_i$ be the size of the set $Q$ considered in the $i$th iteration. Observe the following:
\begin{itemize}
	\item We have $b_i\leq d/2^{i-1}$ for all $i \in [p]$ (because $b_1=d$ and, by Lemma~\ref{claim:runtime}, $b_{i+1}\leq b_i/2$ for all $i\in [p-1]$).
	\item We have  $1\leq \ell_i\leq b_i$ for all $i \in [p]$ (because in the algorithm we consider only nonempty sets $Q$ satisfying $|Q|\leq b$).
	\item We have $\sum_{i=1}^p \ell_i\leq d$ (because $b$ decreases by at least $\ell_i$ in the $i$th iteration, and the procedure terminates once $b$ becomes negative).
\end{itemize}
Therefore, every root-to-leaf path in the recursion tree can be uniquely described by specifying the following data:
\begin{enumerate}[label=(\Alph*)]
 \item\label{ls} Positive integers $\ell_1,\ldots,\ell_p$ satisfying $\sum_{i=1}^p \ell_i\leq d$.
 \item\label{Qs} For each $i\in [p]$, a choice of a subset $Q_i$ of size $\ell_i$ of the set $P_i$, where $P_i,Q_i$ are the sets $P,Q$ considered in the $i$th iteration.
 \item\label{as} For each $i\in [p]$, a choice of symbols guessed to be fixed at the positions of $Q_i$. 
\end{enumerate}
For~\ref{ls}, it is well-known that the number of representations of $d$ as a sum of numbers $\ell_1,\ldots,\ell_p$ is bounded by $2^{\Oh(d)}$. For~\ref{as}, the total number of choices is bounded by
$$\prod_{i=1}^p (|\Sigma|-1)^{\ell_i}\leq (|\Sigma|-1)^d.$$
Finally, for~\ref{Qs} we shall use the following known bound.

\begin{lemma}[cf. Lemma 124 in~\cite{pilipczuk2013tournaments}] 
	\label{claim:inequality}
	If $m,k$ are nonnegative integers, then $\binom{m+k}{k} \leq 2^{2\sqrt{xy}}$.
\end{lemma}
Since we always have $|P_i|\leq 2d$, the number of choices for~\ref{Qs} is bounded as follows:
$$\prod_{i=1}^p \binom{2d}{\ell_i}\leq \prod_{i=1}^p 2^{2\sqrt{2d\ell_i}}\leq \prod_{i=1}^p 2^{2\sqrt{2d\cdot d/2^{i-1}}}=2^{2\sqrt{2}\cdot d\cdot \sum_{i=0}^\infty 2^{-i/2}}=2^{\Oh(d)}.$$
So all in all, the total number of root-to-leaf paths in the recursion tree is bounded by
$$2^{\Oh(d)}\cdot 2^{\Oh(d)}\cdot (|\Sigma|-1)^d=(|\Sigma|-1)^d\cdot 2^{\Oh(d)},$$
as claimed.

It remains to prove the Lemma~\ref{claim:runtime}.

\begin{proof}[Proof of Lemma~\ref{claim:runtime}]
	Let $q_i$ and $q_{i+1}$ be the candidate word respectively at the beginning and at the end of the $i$th iteration, that is, after guessing is performed. 
    
    Recall that there is $s\in \Ss$ with $\dist(s,q_i)>d$. Further, recall that
	$$b_{i+1}=\min(d-\dist(s,q_{i+1}),b_i-|Q|).$$
	To prove that $b_{i+1}\leq b_i/2$, it suffices to show that
	$$(d-\dist(s,q_{i+1}))+(b_i-|Q|)\leq b_i,$$
    or equivalently,
    \begin{equation}\label{eq:goal}
		d\leq \dist(s,q_{i+1})+|Q|.
    \end{equation}
    By triangle inequality we have
	$$d<\dist(s,q_i)\leq \dist(s,q_{i+1})+\dist(q_i,q_{i+1}),$$
    but we also have
	$$\dist(q_i,q_{i+1})=|Q|.$$
    So this establishes~\eqref{eq:goal} and finishes the proof.
\end{proof}

%% file: chapters/meta-theorems.tex
\newcommand{\init}{\mathsf{init}}
\newcommand{\update}{\mathsf{update}}
\newcommand{\query}{\mathsf{query}}

\section{Applications of Meta-Theorems}
\label{sec:meta-theorem}

In this section we first state a meta-theorem for string problems definable in first-order logic $\FO$; this result follows easily from the work of Frandsen et al.~\cite{FrandsenMS97} using the classic Sch\"utzenberger-McNaughton-Papert theorem~\cite{McNaughtonP71,Schutzenberger65a}. Then we explain how to use the meta-theorem for the following toy problems: {\sc{Disjoint Factors}} and {\sc{Edit Distance}}. As usual with meta-theorems, the parametric factor in the complexity guarantees of the obtained data structures is not explicit, and typically is much higher than if one constructs the data structure ``by hand''. Therefore, we next show how to derive concrete data structures with concrete complexity guarantees for {\sc{Disjoint Factors}} and {\sc{Edit Distance}}. Finally, we discuss a methodology for lower bounds introduced by Amarilli et al.~\cite{AmarilliJP21}, and we apply it to derive lower bounds for those two problems.

\subsection{A meta-theorem}

We first need to recall basic knowledge on different equivalent views on regular languages. This material is standard in the area of algebraic theory of languages, so we refer an interested reader to the book of Boja\'nczyk~\cite{Bojanczyk20} for a broader introduction. In particular, we explain the contemporary understanding of the material, and for appropriate references and historical remarks, we refer to~\cite{Bojanczyk20}.

The first view is through the lens of logic.
Fix a finite alphabet $\Sigma$. We consider the logic $\MSO[\Sigma,<]$ operating on words. In this logic there are variables for single positions (denoted with small letters) and subsets of positions (denoted with capital letters). The atomic formulas are of the following form:
\begin{itemize}
 \item equality test $x=y$;
 \item test $x\in X$ checking that position $x$ belongs to position subset $X$;
 \item for every $a\in \Sigma$, test $a(x)$ checking that at position $x$ there is symbol $a$; and
 \item test $x<y$ checking that position $x$ appears before position $y$.
\end{itemize}
Formulas of $\MSO[\Sigma,<]$ can be obtained from atomic formulas using standard boolean connectives and quantification (both universal and existential, and applicable to both types of variables). $\FO[\Sigma,<]$ is a fragment of $\MSO[\Sigma,<]$ where we disallow variables for subsets of positions.

A {\em{sentence}} is a formula without free variables. By $w\models \phi$ we mean that the sentence $\phi$ is satisfied in the word $w$. For a sentence $\phi\in \MSO[\Sigma,<]$, the language {\em{defined}} by $\phi$ consists of all words $w$ in which $\phi$ is satisfied. A language $L\subseteq \Sigma^\star$ is {\em{$\MSO$-definable}} if $L$ is defined by some $\phi\in \MSO[\Sigma,<]$ as above, and {\em{$\FO$-definable}} if it is defined by some $\phi\in \FO[\Sigma,<]$. It appears that regular languages exactly coincide with ones definable in $\MSO$.

\begin{theorem}\label{thm:mso}
 A language of finite words over a finite alphabet is regular if and only if it is $\MSO$-definable.
\end{theorem}

The next view is through semigroup homomorphisms. Consider a language $L\subseteq \Sigma^\star$. By endowing $\Sigma^\star$ with the concatenation operation we can regard it as a semigroup. For another semigroup $S$ and a (semigroup) homomorphism $h\colon \Sigma^\star\to S$, we say that $h$ {\em{recognizes}} $L$ if there exists $A\subseteq S$ such that $L=h^{-1}(A)$; in other words, whether $w\in L$ can be recognized by looking at $h(w)$ and determining whether it belongs to $A$. It turns out that regular languages are also exactly those that are recognized by homomorphisms to finite semigroups.

\begin{theorem}\label{thm:finite-state}
 A language of finite words over a finite alphabet is regular if and only if it is recognized by a homomorphism to a finite semigroup.
\end{theorem}

Further, it is known that if $L$ is regular, then there exists a unique minimal --- in terms of cardinality --- semigroup $S$ such that there is a homomorphism from $\Sigma^\star$ to $S$ recognizing $L$. This semigroup is called the {\em{syntactic semigroup}} for $S$.

It turns out that $\FO$-definable languages can be characterized in terms of algebraic properties of their syntactic semigroups. Here, a semigroup is {\em{aperiodic}} (or {\em{group-free}}) if it does not contain any non-trivial group.

\begin{theorem}[Sch\"utzenberger-McNaughton-Papert Theorem,~\cite{McNaughtonP71,Schutzenberger65a}]\label{thm:schutzeberger}
 A regular language $L$ is $\FO$-definable if and only if its syntactic semigroup is aperiodic.
\end{theorem}

With these standard tools recalled, we can proceed to the setting of dynamic data structures.

\medskip

Fix a finite alphabet $\Sigma$ and consider a language $L\subseteq \Sigma^\star$. The {\em{word problem}} for $L$ is to design a data structure that maintains a dynamic word $w\in \Sigma^\star$ and supports the following operations:
\begin{itemize}
 \item $\init(w)$: Initialize the data structure with the given word $w$.
 \item $\update(i,a)$: Update $w$ by replacing the symbol at position $i$ by symbol $a\in \Sigma$.
 \item $\query()$: Determine whether $w\in L$.
\end{itemize}
The complexity guarantees of such a data structure is typically measured in terms of $n := |w|$. Note that this value is fixed upon initialization and then stays the same throughout the life of the data structure.

We can also consider the word problem for semigroups. Suppose $S$ is a semigroup. Then the word problem for $S$ is defined as above for words over $S$ (that is, words $w\in S^\star$), where query is redefined as follows: Output the (left-to-right) product of all the symbols in $w$.

Observe that the word problem for a regular language $L\subseteq \Sigma^\star$ easily reduces to the word problem for its syntactic semigroup $S$. Indeed, if $h\colon \Sigma^\star\to S$ is the homomorphism recognizing $L$, say $L=h^{-1}(A)$ for some $A\subseteq S$, then in the reduction we can map symbols $a\in \Sigma$ to their images $h(a)\in S$, and whether $w\in L$ can be deduced by checking whether $h(w)\in A$.

Frandsen et al.~\cite{FrandsenMS97} proposed an efficient dynamic data structure for the word problem in aperiodic semigroups.

\begin{theorem}[\cite{FrandsenMS97}]\label{thm:miltersen}
 Let $S$ be a finite aperiodic semigroup. Then there is a data structure for the word problem for $S$ with initialization time $\Oh(n)$, worst-case update time $\Oh(\log \log n)$, and worst-case query time $\Oh(1)$.
\end{theorem}

By combining Theorems~\ref{thm:schutzeberger} and~\ref{thm:miltersen} using the reduction presented above, we obtain the following.

\begin{theorem}\label{thm:meta}
 Let $\Sigma$ be a finite alphabet and suppose $L\subseteq \Sigma^\star$ is $\FO$-definable. Then there is a data structure for the word problem for $L$ with initialization time $\Oh(n)$, worst-case update time $\Oh(\log \log n)$, and worst-case query time $\Oh(1)$.
\end{theorem}

A few remarks are in order. First, the proof of Theorem~\ref{thm:miltersen}
relies on induction on the Khron-Rhodes decomposition of the semigroup $S$,
where in each step of the induction one applies van Emde Boas
trees~\cite{van-emde-boas}. The induction has depth bounded by the size of $S$, so one can view this data structure as $\Oh(|S|)$ van Emde Boas trees stacked ``on top of each other''. Consequently, the constants hidden in the $\Oh(\cdot)$ notation in Theorem~\ref{thm:miltersen} depend on $S$, but not horribly: they are polynomial in $|S|$. However, there is a much more significant complexity blow-up hidden in Theorem~\ref{thm:finite-state}. Specifically, if a regular language $L$ is defined by an $\FO[\Sigma,<]$ sentence $\phi$, then the syntactic semigroup of $L$ has size bounded by a function of $|\phi|$, but this function is in general non-elementary --- it is basically a tower of height equal to the quantifier rank of $\phi$. This non-elementary dependence is known to be unavoidable~\cite{StockmeyerThesis}. Therefore, whenever one applies Theorem~\ref{thm:meta} in order to obtain data structures for a problem based on its description in $\FO$, one should bear in mind that the constants hidden in the $\Oh(\cdot)$ notation depend non-elementarily on the length of the description.

Second, recently Amarilli et al.~\cite{AmarilliJP21} gave a characterization of regular languages for which data structures with guarantees as in Theorem~\ref{thm:miltersen} exist. This characterization renders the tractability region to be a bit broader than just $\FO$-definability, for instance the languages ``on every even position there is symbol $a$'' or ``in total there is an even number of symbols $a$'' are not $\FO$-definable, but admit data structures for the word problem with constant update time. The characterization is expressed in algebraic terms and we could not find natural examples of parameterized string problems that would not be $\FO$-definable, but fall under the characterization. So we refrain from giving more details and point an interested reader to~\cite{AmarilliJP21} instead.


\medskip

We now explain how to use Theorem~\ref{thm:meta} in practice on two examples: {\sc{Disjoint Factors}} and {\sc{Edit Distance}}. In each case, the task boils down to defining the problem in $\FO[\Sigma,<]$ for an appropriate alphabet $\Sigma$.
We start with {\sc{Disjoint Factors}}.

\begin{lemma}\label{lem:df-sentence}
 Let $k\in \N$ and $\Sigma_k=[k]$.
 There is a sentence $\phi_k\in \FO[\Sigma_k,<]$, computable from~$k$, such that for every $w\in \Sigma_k^\star$, $w$ is a yes-instance of {\sc{Disjoint Factors}} for parameter $k$ if and only if $w\models \phi_k$.
\end{lemma}
\begin{proof}
 In the sentence $\phi_k$, we first make a disjunction over all permutations $\pi\colon [k]\to [k]$. For each such $\pi$, we verify that there exist positions $x_1<y_1<x_2<y_2<\ldots<x_k<y_k$ such that for each $i\in [k]$, both at position $x_i$ and at $y_i$ there is symbol $\pi(i)$. It is straightforward to express this condition using an $\FO[\Sigma_k,<]$ sentence.
\end{proof}

By applying Theorem~\ref{thm:meta} to the language defined by sentence $\phi_k$ provided by Lemma~\ref{lem:df-sentence}, we obtain the following.

\begin{corollary}
 There is a data structure for the dynamic {\sc{Disjoint Factors}} problem with initialization time $\Oh_k(n)$, worst-case update time $\Oh_k(\log \log n)$, and query time $\Oh(1)$.
\end{corollary}

 Note here that the query time can be a constant independent of $k$, as we can always recompute the answer to the query following every update.

For {\sc{Edit Distance}}, the formula is more complicated. For two words $u,v\in \Sigma^\star$, by $u\otimes v$ we denote the word over $(\Sigma\cup \{\bot\})^2$, where $\bot$ is a symbol not present in $\Sigma$, defined as follows:
\begin{itemize}
 \item The length of $u\otimes v$ is $\max(|u|,|v|)$.
 \item For each $1\leq i\leq \min(|u|,|v|)$, we put $u\otimes v[i]=(u[i],v[i])$.
 \item For each $\min(|u|,|v|)<i\leq \max(|u|,|v|)$, we put $u\otimes v[i]=(u[i],\bot)$ or $u\otimes v[i]=(\bot,v[i])$, depending on whether $\max(|u|,|v|)=|u|$ or $\max(|u|,|v|)=|v|$.
\end{itemize}

\begin{lemma}\label{lem:ed-sentence}
 Let $k\in \N$ and $\Sigma$ be a finite alphabet.
 There is a sentence $\psi_{k,\Sigma}\in \FO[(\Sigma\cup \{\bot\})^2,<]$, computable from $k$ and $\Sigma$, such that for all $u,v\in \Sigma^\star$, we have $\ED(u,v)\leq k$ if and only if $u\otimes v\models \psi_{k,\Sigma}$.
\end{lemma}
\begin{proof}
 Denote $\Gamma=(\Sigma\cup \{\bot\})^2$ for brevity.
 Note that for two words $u',v'\in \Sigma^\star$ we have $\ED(u',v')\leq k$ if and only if there exist integers $a,b,c\geq 0$ with $a+b+c\leq k$ such that  one can remove $a$ positions from $u'$ and $b$ positions from $v'$ so that the resulting strings have equal length and differ on exactly $c$ positions.
 In such case, we will call the pair $(u',v')$ {\em{$(a,b,c)$-editable}}.
 
 For all $(a,b,c)\in \{0,1,\ldots,k\}^3$ with $a+b+c\leq k$ and $s,t\in \{-k,\ldots,k\}$, we shall construct a formula $\alpha_{s,t,a,b,c}(x,y)$ that satisfies the following: for two positions $1\leq x\leq y\leq \max(|u|,|v|)$, we have
 $$u\otimes v\models \alpha_{s,t,a,b,c}(x,y)\qquad\textrm{if and only if}\qquad  (u[x : y],v[x+s : y+t])\textrm{ is }(a,b,c)\textrm{-editable.}$$
 Here, we use the convention that if the specified range $[x : y]$ or $[x+s : y+t]$ does not fit into the corresponding word, or makes no sense due to $x>y+1$ or $x+s>y+t+1$, then $\alpha_{s,t,a,b,c}(x,y)$ should be false (this can be easily recognized in $\FO[\Gamma,<]$). If we achieve the above, the formula $\psi_{k,\Sigma}$ can be defined as the disjunction of all formulas $\alpha_{0,t,a,b,c}(1,n_u)$ for $a,b,c$ as above, where $1$ is the first position, $n_u$ is the last position of $u$, and $t\in \{-k,\ldots,k\}$ is such that $n_u+t$ is the last position of~$v$ (all these are easily definable from $u\otimes v$ in $\FO[\Gamma,<]$).
 
 The construction is by induction on $a+b$. For the base case $a=b=0$, we may
 define $\alpha_{s,t,0,0,c}(x,y)$ as follows: if $s\neq t$ then the formula is
 always false, and otherwise it checks whether there are exactly $c$ different
 positions $z$ such that $x\leq z\leq y$ and $u[z]\neq v[z+s]$. This can be
 checked by comparing the first coordinate of $u\otimes v[z]$ with the second
 coordinate of $u\otimes v[z+s]$. Since $|s|\leq k$ by assumption, it is straightforward to formulate this assertion in $\FO[\Gamma,<]$.
 
 We proceed to the induction step. So assume $a+b>0$, say $a>0$; the construction in the case $b>0$ is analogous, so we omit it. The idea is that we guess, by existential quantification, the first position in $u[x : y]$ that gets removed, and use simpler formulas given by the induction assumption. More precisely, $\alpha_{s,t,a,b,c}(x,y)$ can be defined as the conjunction of formulas
 $$\exists z.\left(x\leq z\leq y \wedge \alpha_{s,r,0,b_1,c_1}(x,z-1)\wedge \alpha_{r-1,t,a-1,b_2,c_2}(z+1,y)\right),$$
 for all integers $b_1,b_2\geq 0$ with $b_1+b_2=b$, $c_1,c_2\geq 0$ with $c_1+c_2=c$, and $r\in \{-k+1,\ldots,k\}$. Here, $z-1$ and $z+1$ are a syntactic sugar for the predecessor and the successor of $z$, respectively, which are easily definable in $\FO[\Gamma,<]$.
 It is straightforward to see that the construction of $\alpha_{s,t,a,b,c}(x,y)$ as above satisfies the required properties.
\end{proof}

Similarly as before, by combining Theorem~\ref{thm:meta} with Lemma~\ref{lem:ed-sentence} we obtain the following.

\begin{corollary}
 There is a data structure for the dynamic {\sc{Edit Distance}} problem with initialization time $\Oh_{k,\Sigma}(n)$, worst-case update time $\Oh_{k,\Sigma}(\log \log n)$, and query time $\Oh(1)$. 
\end{corollary}

%% file: chapters/examples.tex
\section{Improved data structures}
\label{sec:examples}

In the Section~\ref{sec:meta-theorem}, we showed that powerful meta-theorem seamlessly
guarantee $\Oh_k(\log\log n)$ worst-case update time for dynamic versions of {\sc{Disjoint
Factors}} and {\sc{Edit Distance}}. Now, we illustrate that the dependence on
$k$ can be significantly improved by exploiting combinatorial structure of this
problems. In the Section~\ref{sec:df} we give an improved dependence for
{\sc{Disjoint Factors}} and in Section~\ref{sec:edit-distance} we present a data structure for {\sc{Edit Distance}}.

\subsection{$k$-Disjoint Factors}
\label{sec:df}

\begin{lemma}
    \label{lem:dynamic-df}
     There is a data structure for the dynamic {\sc{Disjoint Factors}} problem
     with initialization time $\Oh(kn)$, worst-case update time $\Oh(\log\log n)$, and query time $\Oh(k 2^{k} \log\log n)$.
\end{lemma}

As a by product we also present an improved $\Oh(k 2^k + kn)$ time algorithm for
static version of {\sc{Disjoint Factors}} problem. This improves an $\Oh(k 2^k
n)$ time algorithm due to Bodlaender et al.~\cite{bodleander11}.

\begin{corollary}
    \label{cor:static-df}
    {\sc{Disjoint Factors}} can be statically solved in $\Oh(k 2^k + kn)$ time.
\end{corollary}

We calculate the solution to the {\sc{Disjoint Factors}} with dynamic
programming. The entries of the dynamic programming are parameterized by sets $S
\subseteq \Sigma$. For each entry, we store a minimal integer $\ell \in [n]$,
such that a word $w[1 : \ell]$ can be represented as disjoint factors with
letters from $S$. 

We compute the entries of the dynamic programming in the bottom-up fashion. In
the base case $S = \emptyset$ the answer is $\ell = 1$. When $S$ is nonempty, we
guess a letter $s \in S$. Next, we determine the position of $\ell_s$ for set
$S\setminus\{s\}$. Then, we look for the first $s$-factor after position
$\ell_s$ in the word and return the position of its right-endpoint. We summarize
is as follows:

\begin{displaymath}
    {\mathtt{DF}}[S] := \begin{cases}
        \min_{s \in S} \left\{ \text{next } s \text{-factor after } \mathtt{DF}[S \setminus \{s\}]\right\} & \text{if } S \neq \emptyset,\\
        1 & \text{otherwise}.
    \end{cases}
\end{displaymath}

For correctness, observe that every optimal solution to {\sc{Disjoint Factor}}
can be represented as a permutation $\pi : [k] \rightarrow [k]$ that
corresponds to the order in which each factor appears in the word. Moreover, if
the solution to {\sc{Disjoint Factor}} that is represented by permutation
$\pi$ exists, then it can be detected by greedily taking subsequent factors
in the word (cf.,~\cite{bodleander11}).

Now, we focus on the exact implementation of determining a next factor after
position $\ell$.

\begin{proof}[Proof of Corolary~\ref{cor:static-df}]
    It remains to show that there exists a data-structure that can be
    initialized in  $\Oh(kn)$ time
    and can answer queries of the form $(s,\ell) \in \Sigma \times [n]$. The
    query returns a position of the first right-endpoint of
    $s$-factor in the word $w[\ell:n]$. To achieve this, we store the table
    $\mathtt{next}[s,\ell]$ for every $s \in \Sigma_k$ and $\ell \in
    [n]$. It stores the minimal position of a letter $s \in \Sigma_k$ in the
    word $w[\ell:n]$. To compute $\mathtt{next}[s,\ell]$ for a fixed $s
    \in \Sigma_k$, we fill it starting from $\ell$ equal to $n$ down to
    $1$. Initially $\mathtt{next}[s,n] = \infty$ for every $s \in
    \Sigma_k$. If $w[\ell] = s$, then we set $\mathtt{next}[s,\ell] = \ell$.
    Otherwise, we know that the next position of a letter $s$ is after
    $\ell$ and we set the value of the table $\mathtt{next}$ at $s$ and
    $\ell$ to $\mathtt{next}[s,\ell+1]$.

    Therefore, to find the position of next $s$-factor after position $\ell$ we lookup
    the value of $\ell' := \mathtt{next}[s,\ell]$. This means that next
    $s$-factor starts at position $\ell'$. We return the position of
    $\mathtt{next}[s,\ell']$ as the right-endpoint of this
    $s$-factor.

    Observe, that we can compute table $\mathtt{next}$ in time
    $\Oh(kn)$. Moreover, each $\mathtt{next}$ query takes $\Oh(1)$ time. Hence, with
    $\mathtt{DF}$ procedure we can solve \textsc{Disjoint
    Factor} in $\Oh(k 2^k + kn)$ time.
\end{proof}

\begin{proof}[Proof of Lemma~\ref{lem:dynamic-df}]
    In the dynamic setting, we need to show that for every $s \in
    \Sigma_k$ and $\ell \in [n]$ the position of a next $s$-factor after
    position $\ell$ can be found in $\Oh(\log\log n)$ time.  To achieve that,
    for every letter $s \in \Sigma_k$ we maintain a van Emde Boas tree $\Tt_s$.
    In the data structure, $\Tt_s$ stores the positions of
    a letter $s$ in the updated word. Note, that a single update to $\Tt_s$ can be
    done in $\Oh(\log\log n)$ time. During query, for a given $s \in \Sigma_k$ and $\ell
    \in [n]$ we can ask for the next position of a letter $s$ after $\ell$. A
    single query to van Emde Boas tree takes $\Oh(\log\log n)$ times.
    
    Similarly to the proof of Corollary~\ref{cor:static-df}, we can use $\Tt_s$
    to find a position of $s$-factor after a given position in
    $\Oh(\log\log n)$, with two queries to $\Tt_s$. Observe, that during an
    update to {\sc{Disjoint Factor}} we need to update only two van Emde Boas
    trees. Hence update takes $\Oh(\log\log n)$ time. Moreover, $\mathtt{DF}$
    procedure can answer queries to {\sc{Disjoint Factor}} in $\Oh(k 2^k
    \log\log n)$ time. 
\end{proof}

\subsection{Edit Distance}
\label{sec:edit-distance}

\newcommand{\LCE}{\mathsf{LCE}}

\begin{lemma}
    \label{lem:ed}
     There is a data structure for the dynamic {\sc{Edit Distance}} problem
     with initialization time $\Oh(kn)$, worst-case update time $\Oh(k \log\log n)$,
     and query time $\Oh(k^2 \log\log n)$.
\end{lemma}

The data structure is based on the classical result of Landau and
Vishkin~\cite{landau-vishkin} who gave a static $\Oh(n+k^2)$ time algorithm for
$k$-\textsc{Edit Distance} problem. Our simple observation is that $\LCE$-queries in
their algorithm can be efficiently maintained with van Emde Boas trees.

Before we proceed with the proof of Lemma~\ref{lem:ed} let us recall the
definition of \emph{Longest Common Extension} query. Let $x,y$ be two strings of
length at most $n$. We define $\LCE(i,j)$ as the largest integer $\ell \in [n]$,
such that $x[i:i+\ell] = y[j:j+\ell])$.  We show, that with van Emde Boas data
structure, we can efficiently maintain an answer to the $\LCE$ queries when
$|i-j|$ is small.

\begin{proposition}[Dynamic $\LCE$-queries]
    \label{prop:lce}
    For any $k \in \nat$, there is a data structure that maintains words $x,y \in \Sigma^n$ and
    supports the following operations:
    \begin{itemize}
        \item $\init(x,y)$: Initialize the data structure with the given words $x,y$,
        \item $\update(w,i,s)$: Update a word $w \in \{x,y\}$ by replacing the symbol
            at position $i$ with $s\in \Sigma$,
        \item $\query(i,j)$: if $|i-j| \le k$ return $\LCE(i,j)$ of $x,y$.
    \end{itemize}
    The $\init$ operation requires $\Oh(kn)$ worst-case time. Operations
    $\update$ and $\query$ can be executed in worst-case
    time $\Oh(k \log\log n)$ and $\Oh(\log\log n)$, respectively.
\end{proposition}

\begin{proof} 
    For every fixed $p \in \{-k,\ldots,k\}$ we maintain the van Emde
    Boas data-structure $\Tt_p$. Each $\Tt_p$ maintains a set $S_p
    \subseteq [n]$ defined as:

    \begin{displaymath}
        S_p := \{ i \in [n] \text{ such that } x[i] \neq y[i+p]\}
        .
    \end{displaymath}
    
    Observe, that sets $S_p$ can be efficiently maintained with van Emde Boas
    data structures. During each update to word $x[i]$ we iterate over
    every $p \in \{-k,\ldots,k\}$, check if $x[i] \neq y[i+p]$ and update
    $\Tt_p$ accordingly. Single query to $\LCE$ can be done with a single query
    to $\Tt_p$.
\end{proof}

With that data-structure to answer $\LCE$-queries under text updates we can proceed with
the description of the dynamic data structure for $k$-\textsc{Edit Distance} with
$\Oh(k \log\log n)$ worst-case update time.

\begin{proof}[Proof of Lemma~\ref{lem:ed}]

    When length of $x$ and $y$ differ by more than $k$ then the edit distance
    between $x$ and $y$ must be greater than $k$ and we can conclude that the
    answer is negative.  Our algorithm is based on the static algorithm of
    Landau and Vishkin~\cite{landau-vishkin}. They consider the dynamic
    programming table $\mathtt{D}[i,j] := \mathtt{ed}(x[1:i],y[1:j])$ for edit
    distance. They claim that it suffices to touch only $\Oh(k^2)$ entries of
    the table $\mathtt{D}$ in order to retrieve $\mathtt{D}[|x|,|y|]$ (if its
    value is $\Oh(k)$). 

    The main observation of the static $\Oh(n+k^2)$ time algorithm of Landau and
    Vishkin~\cite{landau-vishkin} is that $\mathtt{D}[i+1,j+1] \in
    \{\mathtt{D}[i,j], \mathtt{D}[i,j]+1\}$. Therefore, it is always beneficial
    to greedily take $\LCE$ to compute the next value in the table $\mathtt{D}$.
    Because we assumed that the edit distance between $x$ and $y$ is at most
    $k$, we need to perform at most $k$ queries to $\LCE$. Note that we consider
    only values $\mathtt{D}[i,i+p]$ for $p \in \{-k,\ldots,k\}$. Hence we need
    to touch at most $2k+1$ entries with value exactly $\ell$ for every $\ell
    \in [k]$. In total number of entries in table $\mathtt{D}$ touched by Landau
    and Vishkin algorithm is $\Oh(k^2)$.

    It remains to augment this static algorithm into the dynamic setting. We
    keep track of $\Oh(k)$ data structures from Proposition~\ref{prop:lce} to
    emulate $\LCE$ queries (each data-structure corresponds to one diagonal $p
    \in \{-k,\ldots,k\}$ of Landau-Vishkin algorithm). Single update may require
    modification to $\Oh(k)$ of this data structures which result in $\Oh(k
    \log\log(n))$ update time.  Single query to \textsc{Edit Distance} problem
    takes $\Oh(k^2 \log\log n)$ time because we Vishkin and Landau algorithm
    performs $\Oh(k^2)$ queries to $\LCE$ (each can be done in
    $\Oh(\log\log(n))$ time with a data structure from
    Proposition~\ref{prop:lce}).
\end{proof}

%% file: chapters/lower-bounds.tex
\section{Lower bounds}
\label{sec:lower-bounds}
\newcommand{\PU}{\textsc{prefix}-$U_1$\xspace}

In the \PU problem the task is to maintain a subset $S$ of the universe
$\{1,\ldots,n\}$ under deletions and insertions, and support \emph{threshold
queries}: given $i \in [n]$, decide whether $S$ contains some element that is
$\le i$. The \PU problem can be solved in deterministic $\Oh(\log\log(n))$ time
with a predecessor search~\cite{van-emde-boas,PatrascuT14} or in expected
$\Oh(\sqrt{\log\log{n}})$ time if randomization is allowed~\cite{prefix-u1-alg1,prefix-u1-alg2}. (Here, by time we mean the worse of the update and query times.) Unfortunately, no unconditional
lower bounds is known for \PU. 

Amarilli et al.~\cite{AmarilliJP21} presented a large class of problems that are
at least as hard as the \PU problem. Moreover, they showed that there are problems
equivalent to \PU. These led them to conjecture that there is no data structure for \PU that offers updates and queries in $\Oh(1)$.

\begin{conjecture}[\cite{AmarilliJP21}]
    \label{conj:pu}
    There is no data structure for \PU that achieves $\Oh(1)$ amortized time for updates and queries  in the word-RAM model.
\end{conjecture}

For large enough word-size $w$ the complexity of even harder
{\sc{Predecessor}} is $\Oh(1)$ in the word-RAM model~\cite{FredmanW93}. However, for
general $w$ and $n$, there are tight lower bounds~\cite{PatrascuT14}. In
Conjecture~\ref{conj:pu} we just expect that the complexity of \PU cannot be
$\Oh(1)$ for general $w$ and~$n$.

We say that a problem $\mathcal{P}$ is {\em{\PU-hard}} if assuming that $\mathcal{P}$ admits a data structure achieving amortized $\Oh(1)$ time of operations, the same can be said also about \PU. In this section we show that the dynamic versions of \textsc{Disjoint Factors} and of
\textsc{Edit Distance} are \PU-hard.

\subsection{Lower bound for \textsc{Disjoint Factors}}

Now we prove that the $k$-\textsc{Disjoint Factors} problem is \PU-hard for $k=3$. 
We complement our result and show that \textsc{Disjoint Factors} for $k=2$
admits a data structure with operations taking $\Oh(1)$ time.

\begin{lemma}
     Unless Conjecture~\ref{conj:pu} fails, there does not exist a data structure for the dynamic variant of \textsc{Disjoint Factors} for $k \ge 3$ with $\Oh(1)$ amortized 
    update and query time.
\end{lemma}
\begin{proof} 
    We show a reduction from \PU to \textsc{Disjoint Factors} on alphabets of
    size $3$. Without loss of generality let us assume that $\Sigma \coloneqq 
    \{\mathtt{0,1,\#\}}$. Assume that $S \subseteq [n]$ is the set maintained in the \PU
    problem. Based on $S$ we construct a word:
    \begin{displaymath}
        w(S) := \mathtt{1} \; \alpha_1\; \alpha_2 \; \ldots \;\alpha_n \; \mathtt{0\;\#\;0\;0},
    \end{displaymath}
    where $\alpha_i$ is $\mathtt{1}$ if $i \in S$ and $\mathtt{0}$ otherwise. On $w(S)$ we maintain the assumed data structure for \textsc{Disjoint Factors}.
    The updates to set $S$ are relayed to $w(S)$ in a natural manner.
    For a threshold query to \PU of form $i \in [n]$ we do the following:
    \begin{itemize}
     \item Temporarily change the symbol
    at position $(i+2)$ of $w(S)$ to $\mathtt{\#}$.
     \item Query the data structure for \textsc{Disjoint Factors} to decide whether the current word is a positive instance, and return this answer as the answer to the query.
     \item Finally, revert the word back to the original $w(S)$. 
    \end{itemize}
     This concludes
    the construction of the reduction. Observe, that each step of that reduction
    can be implemented in (amortized) $\Oh(1)$ time.

    For correctness, observe that after a query $i \in [n]$ the input to
    \textsc{Disjoint Factors} is of the form (when $i < n$):
    \begin{displaymath}
        \mathtt{1} \; \alpha_1  \ldots \alpha_i \; \mathtt{\#} \; \alpha_{i+2} \ldots
        \alpha_n \; \mathtt{0\;\#\;0\;0}.
    \end{displaymath}
    In the resulting word, there are only two occurrences of symbol
    $\mathtt{\#}$. Therefore, if the answer to \textsc{Disjoint Factors} is
    positive, then word $\mathtt{\#} \alpha_{i+2} \ldots \alpha_n \mathtt{\#}$ is a $\mathtt{\#}$-factor. Observe that two last letters
    $\mathtt{0 \; 0}$ already form a $\mathtt{0}$-factor. We are left with
    determining the position of a $\mathtt{1}$-factor. However, all
    factors need to be pairwise disjoint. Hence $\mathtt{1}$-factor exists iff at least one of
    $\alpha_1,\ldots,\alpha_i$ is $\mathtt{1}$. This is possible only if $[i]
    \cap S \neq \emptyset$, which proves the correctness of our reduction.
\end{proof}

Now we show that \textsc{Disjoint Factors} for $k=2$ can be
maintained in $\Oh(1)$ time. We first need a simple lemma.

\begin{lemma}
    \label{lemma:k1color}
    For any word $w \in \Sigma^\star$, if every symbol in $w$ occurs more than $|\Sigma|$ times,
    then the answer to \textsc{Disjoint Factors} on $w$ is positive.
\end{lemma}

\begin{proof}
    Let $k := |\Sigma|$. We prove the statement by induction on $k \ge 1$.
    In the base case $k=1$ the statement trivially holds. 
    Let $w \in \Sigma^\star$ be a word over alphabet of size $k$ were each letter appears at least $k+1$ times.
    Let $i \in [n]$ be the maximal index such that no symbol from $\Sigma$ is repeated in $w[1:i]$.
    It follows that in $w[i+1:n]$ each symbol occurs at least $k$ times.
    By the maximality of $i$ the symbol $\alpha := w[i+1]$ occurs twice in $w[1:i+1]$.
    Let us greedily take these two occurrences of $\alpha$ as an $\alpha$-factor.

    Consider the word $w[i+1:n]$ and delete from it all
    occurrences of $\alpha$. This word contains $k-1$ symbols and each symbol
    occurs at least $k$ times. Therefore, by the induction assumption it contains $k-1$
    disjoint factors. By combining them with our $\alpha$-factor, we see that the
    original word $w$ contains $k$ disjoint factors.
\end{proof}

Now, we show that \textsc{Disjoint Factors} for binary alphabets can be
dynamically maintained in $\Oh(1)$ time. 

\begin{lemma}
    There exists a data-structure for the dynamic variant of \textsc{Disjoint Factors} for $k=2$  that supports updates and queries in amortized time $\Oh(1)$.
\end{lemma}
\begin{proof}
    Consider Lemma~\ref{lemma:k1color} for an alphabet of size $2$, say $\Sigma=\{\mathtt{0},\mathtt{1}\}$. If both
    symbols occur at least $3$ times in a given word, then the answer to
    \textsc{Disjoint Factors} is positive. On the other hand, if one symbol
    appears less than twice, then the answer to \textsc{Disjoint
    Factors} must be negative.
    
    It remains to consider the situation when one of the symbols occurs exactly
    twice.  Let $w \in \Sigma^n$ be the considered word
    Without loss of generality let $\mathtt{1}$ be the symbol that occurs twice
    in $w$ on positions $i,j \in [n]$. Therefore, we need to take these
    positions to construct a $\mathtt{1}$-factor.  Observe, that if an answer to
    \textsc{Disjoint Factors} is positive then the $\mathtt{0}$-factor must be
    contained in either $w[1,i-1]$ or $w[j+1:n]$. Since the symbol $\mathtt{1}$
    appears only twice, both of these words consist of symbol $\mathtt{0}$
    exclusively. Hence the answer to \textsc{Disjoint Factors}
    is negative iff $i \le 2$ and $j \ge n-1$.

    Necessary information to store and verify this procedure can be easily maintained in $\Oh(1)$ update/query time. 
\end{proof}

\subsection{Lower bound for \textsc{Edit Distance}}

Next, we prove a lower bound for \textsc{Edit Distance} problem for $k=2$.  In
this section we consider words over alphabet $\Sigma=\{\mathtt{a,b,c}\}$. Before
we proceed let us consider a simple gadget.

\begin{claim}
    \label{claim:gadget}
    Consider any word $w \in \{\mathtt{a,b}\}^n$. Then the edit distance
    between words
    \begin{displaymath}
        \mathtt{aca}\, w\, \mathtt{cac}\qquad \text{ and }\qquad \mathtt{cac}\, w\, \mathtt{aca}
    \end{displaymath}
    is equal to $2$ if $w = \mathtt{a}^n$. Otherwise, if $w$ contains any symbol $\mathtt{b}$, then the edit distance is strictly greater than $2$.
    \label{lemma:1b1Wb1b}
\end{claim}

We postpone the proof of Claim~\ref{claim:gadget} for a moment. Let us now use it to prove the \PU-hardness of \textsc{Edit Distance} for
distances at least $2$.

\begin{lemma}
Unless Conjecture~\ref{conj:pu} fails,
    there does not exist a data structure for the dynamic variant of \textsc{Edit Distance} for $k \ge 2$ that supports updates and queries in amortized time $\Oh(1)$.
\end{lemma}

\begin{proof}
    We reduce \PU to $2$-\textsc{Edit Distance} on $\Sigma = \{ \mathtt{a,b,c}\}$.
    Let $S \subseteq [n]$ be the set maintained in the \PU problem.
    Based on $S$ we maintain the following input words to the \textsc{Edit Distance}
    problem:
    \begin{align*}
        x_S := & \,\mathtt{a\,c\,a}\, \alpha_1\,\ldots\,\alpha_n\,
        \mathtt{a\,a\,a} \\
        y_S := & \,\mathtt{c\,a\,c}\,\alpha_1\,\ldots\,\alpha_n \,\mathtt{a\,a\,a}
    \end{align*}
    where $\alpha_i \coloneqq  \mathtt{b}$ if $i \in S$ and $\alpha_i \coloneqq \mathtt{a}$ otherwise. Updates to $S$ are naturally relayed to $x_S$ and $y_S$.

    During a query $i \in [n]$ to the \PU problem we make the following
    modifications in $x_S$ and $y_S$. First, we replace symbols in
    $x_S[i+3:i+5]$ with $\mathtt{cac}$ and in $y_S[i+3:i+5]$ with $\mathtt{aca}$.
    After the modification our two words look as follows (when $i < n-3$):
    \begin{align*}
         & \,\mathtt{a\,c\,a}\, \alpha_1\,\ldots\,\alpha_i\, \mathtt{c\,a\,c} \, \alpha_{i+4}\,\ldots\,\alpha_n \mathtt{a\,a\,a} \\
         & \,\mathtt{c\,a\,c}\, \alpha_1\,\ldots\,\alpha_i\, \mathtt{a\,c\,a} \, \alpha_{i+4}\,\ldots\,\alpha_n \mathtt{a\,a\,a}
    \end{align*}
    Next we query the data structure for \textsc{Edit Distance} to find out whether the edit distance between these two words is larger than $2$. If so, we conclude that the answer
    to \PU query is positive, and otherwise it is negative. Then,
    we clean up and restore the words to the original form of $x_S$ and $y_S$.
    This concludes the description of our reduction. Observe that all the steps
    can be done in (amortized) $\Oh(1)$ time.

    For the correctness, observe that during a query, the input words to \textsc{Edit
    Distance} are of the following form: $\mathtt{aca}\, w\, \mathtt{cac}\, v$ and 
    $\mathtt{cac}\, w\, \mathtt{aca}\, v$, for some
    $w,v \in \{\mathtt{a,b}\}^\star$.  By
    Claim~\ref{claim:gadget} the edit distance between these two words is $2$ iff
    $w = \mathtt{a}^\star$. This happens iff $S
    \cap [i] = \emptyset$, which concludes the correctness proof.
\end{proof}

The remaining proof of Claim~\ref{claim:gadget} can be shown by a diligent case
analysis.

\begin{proof}[Proof of Claim~\ref{claim:gadget}]
    Observe that the edit distance between words is at least $1$. Indeed, on one
    hand the words have equal length hence only a single substitution is possible. On the other
    hand, the Hamming distance between words is $6$ so a single substitution operation
    is not~sufficient.

    Hence, we need to verify when it is possible to transform the first word into the second word
    by using two edit operations. Because the words have an equal length and
    the Hamming distance between them is $6$, two substitutions are not sufficient.
    Therefore, only insertions and deletions are possible. We cannot perform
    just two insertions or just two deletions to the first word, because the
    words have the same length. Therefore, we can only do one
    insertion and one deletion.

    Observe that initially the number of symbols $\mathtt{a}$, $\mathtt{b}$, and
    $\mathtt{c}$ is the same in the both words. Therefore, the letter we insert and the letter we delete must be the same letter. On the other hand, the words differ at the first and
    the last position. Hence, we can only (a) insert $\mathtt{c}$ to the first
    position of the first word and delete $\mathtt{c}$ from the last position of
    the first word; or (b) delete $\mathtt{a}$ from the first position of the first word and insert $\mathtt{a}$ to the last position of the first word.
    Now, observe that (i) if $w = \mathtt{a}^n$ then the
    words are the same after applying (a), and (ii) if there is at least one symbol $\mathtt{b}$ in
    $w$, then regardless whether we apply (a) or (b), the operations cause a misalignment at the position of $\mathtt{b}$s
    and the words differ.
\end{proof}

%% file: chapters/appendix.tex
\section{Omitted proofs}
\label{sec:maintain-phi}

\begin{proof}[Proof of Lemma~\ref{lem:maintain-phi}]
    We assume the familiarity with the proof of Lemma~\ref{lem:base-lemma}, as
    our data structure will extend the one proposed there. Recall that in the context of Lemma~\ref{lem:maintain-phi}, we maintain the data structure of Lemma~\ref{lem:base-lemma} for the dictionary $\widetilde{\Ss}$ and $o$ is the maintained word.
    
    The data structure of Lemma~\ref{lem:base-lemma} stores, for every $\tilde{s} \in \widetilde{\Ss}$, the following set:
    \begin{displaymath}
        \Delta(o,\tilde{s}) = \{ i \in [L] \; | \; o[i] \neq \tilde{s}[i] \}.
    \end{displaymath}
    As, we have a fixed coloring $\pi \colon [L] \to
    [16d]$, we can maintain a table of counters $$\Tt(\tilde{s},c) \coloneqq |\{ i \in [L] \; | \; \pi(i)
    = c \text{ and } i \in \Delta(o,\tilde{s}) \}|\qquad\textrm{for all }\tilde{s} \in
    \widetilde{\Ss}\textrm{ and }c \in [16d].$$ Upon initialization, we explicitly compute
    $\Tt(\tilde{s},c)$ for every $\tilde{s} \in \tilde{S}$ and $c \in [16d]$ in
    $\Oh(n L + nd)$ time.
    During an update to the position $i \in [L]$ of a word $\tilde{s} \in
    \widetilde{\Ss}$ we check if
    $o[i] \neq \tilde{s}[i]$ and update the number of colors in $\Tt(\tilde{s},c)$ based on $\pi(i)$
    accordingly. It may happen that the update to $\tilde{s}$ at position $i$ triggered
    a modification of the word $o$ at position $i$. Then we need to
    change $\Tt(\tilde{s},\pi(i))$ for every $\tilde{s} \in \widetilde{\Ss}$. Note
    that this alone requires $\Oh(|\Ss|)$ time. However, recall that in the proof of
    Lemma~\ref{lem:base-lemma} we argued that before update when the position $i$ of the word $o$ changes, there were at least $|\widetilde{\Ss}|/4$ updates to that position where $o$ was not modified
    (recall here that we work over the binary alphabet).  Therefore, as in the proof of Lemma~\ref{lem:base-lemma}, we may charge the running time $\Oh(|\Ss|)$ to those previous updates to argue that the amortized update time is $\Oh(1)$.  
    
    Observe
    that based on the table $\Tt$, we can  compute the set
    $\colors_{o,\pi}(\tilde{s}) = \{ \pi(i) \; | \; i \in [L] \text{ and }
    \tilde{s}[i] = o[i]\}$ for any given $\tilde{s}$ in time $\Oh(d)$, because it is enough to iterate
    through all $c \in [16d]$ and check whether $\Tt(\tilde{s},c)>0$.

    Now we describe how to maintain the sets $$\Phi(C) = \{\tilde{s} \in \widetilde{\Ss}
    \; | \; \colors_{o,\pi}(\tilde{s}) = C\}\qquad \textrm{for every }C\subseteq [16d].$$
    Each set $\Phi(C)$ is stored as a doubly-linked list of pointers to words from
    $\widetilde{\Ss}$. Upon initialization, we iterate over every $\tilde{s} \in
    \tilde{S}$, lookup the value of $C_{\tilde{s}} \coloneqq \colors_{o,\pi}(\tilde{s})$
    and add a pointer to $\tilde{s}$ to the list $\Phi(C_{\tilde{s}})$.
    Observe, that this operation can be implemented in total time $2^{\Oh(d)}+n L$ time, as we
    can compute sets $C_{\tilde{s}}$ for all $\tilde{s} \in \widetilde{\Ss}$ in total time
    $\Oh(nL)$.

    Next, when a word $\tilde{s}$ is updated on some position and
    is changed to $\tilde{s}'$, we compute the previous set of colors $C \coloneqq \colors_{o,\pi}(\tilde{s})$ and
    the new set of colors $C' \coloneqq \colors_{o,\pi}(\tilde{s'})$. As argued, this operation can be
    done in $\Oh(d)$ time. Next, we delete the pointer to the
    word $\tilde{s}$ from the list $\Phi(C)$ and append a pointer to
    $\tilde{s}'$ to list $\Phi(C')$. Alongside $\tilde{s}$ we store the
    pointer to its list entry in the list $\Phi(C')$ in order to be able to remove it efficiently.
    Both of these operations can be implemented in $\Oh(1)$ time.
    During a query $C \subseteq [16d]$ we return any element from the list $\Phi(C)$
    or assert that it is empty in $\Oh(1)$ time.
\end{proof}